%% file: ms.tex
\documentclass[journal]{IEEEtran}
\usepackage{mathtools}
\usepackage{enumitem}
\usepackage[all]{xy}
\usepackage{algorithm}
\usepackage{algpseudocode}
\usepackage{soul}

\input{edits/edits.tex}

\newcommand{\smax}{s_{\max}}

\newtheorem{test}{Test}
\newtheorem{ass}{Assumption}

\newtheorem*{property}{(P)}

\ifCLASSINFOpdf
\else
\fi
\ifCLASSOPTIONcompsoc
 \usepackage[caption=false,font=normalsize,labelfont=sf,textfont=sf]{subfig}
\else
 \usepackage[caption=false,font=footnotesize]{subfig}
\fi
\begin{document}
%





\title{What is the Largest Sparsity Pattern that Can Be Recovered by 1-Norm Minimization?}

%
%

\author{Mustafa~D.~Kaba, Mengnan~Zhao, Ren\'{e}~Vidal, Daniel~P.~Robinson, Enrique~Mallada
\thanks{{A preliminary version of this paper was presented}
at the $57^{th}$ IEEE Conference on Decision and Control \cite{ZhKaViRoMa:18}.}
\thanks{This work was supported by NSF grants CCF~1618637, IIS~1704458, AMPS~1736448 and CAREER~1752362.}
\thanks{M.D. Kaba is with the Department of Applied Mathematics and Statistics, M. Zhao and E. Mallada are with the Department of Electrical and Computer Engineering, and R. Vidal is with the Mathematical Institute for Data Science and the Department of Biomedical Engineering, at Johns Hopkins University, Baltimore, MD 21218, USA. D.P. Robinson is with the Department of Industrial and Systems Engineering at Lehigh University, Bethlehem, PA 18015, USA. Emails: {\tt \small \{mkaba1, mzhao21, rvidal, mallada\}@jhu.edu} and {\tt \small dpr219@lehigh.edu}.}
}

\maketitle


\begin{abstract}
Much of the existing literature in sparse recovery is concerned with the following question: given a sparsity pattern and a corresponding regularizer, derive conditions on the dictionary under which exact recovery is possible. In this paper, we study the opposite question: given a dictionary and the $\ell_1$-norm regularizer, find the largest sparsity pattern that can be recovered. We show that such a pattern is described by a mathematical object called a ``maximum abstract simplicial complex," and provide two different characterizations of this object:
one based on extreme points and the other based on vectors of minimal support. In addition, we show how this new framework is useful in the study of sparse recovery problems when the dictionary takes the form of a graph incidence matrix or a partial discrete Fourier transform.  In case of incidence matrices, we show that the largest sparsity pattern that can be recovered is determined by the set of simple cycles of the graph. As a byproduct, we show that standard sparse recovery can be certified in polynomial time, although this is known to be NP-hard for general matrices. In the case of the partial discrete Fourier transform, our characterization of the largest sparsity pattern that can be recovered requires the unknown signal to be real and its dimension to be a prime number.
\end{abstract}


%
\IEEEpeerreviewmaketitle

\section{Introduction}
%
%
%
%
\IEEEPARstart{T}{he} widespread use of sparse recovery methods \cite{donoho2006compressed,mallat2008wavelet} in data acquisition \cite{candes2008introduction}, machine learning \cite{Elhamifar:NIPS11,Elhamifar:NIPS12,Elhamifar:TPAMI13,wright2009robust}, medical imaging \cite{lustig2007sparse,candes2006robust,LiMiZoSe:15,Schwab:SIIMS19}, and networking \cite{coates2007compressed, haupt2008compressed, xu2011compressive} has made sparse recovery a popular research area and application tool. The goal of sparse recovery is to find an unknown signal $\bar{x}\in \mathbb{R}^n$ from {a number of observations} $\Phi \bar{x}$, where $\Phi \in \mathbb{R}^{m\times n}$ is a \emph{dictionary} or \emph{measurement matrix} {possibly (but not necessarily)} with $m \ll n$, under the assumption that the unknown signal $\bar{x}$ is sparse in some sense. For instance, \emph{regular sparsity} assumes that the number of nonzero entries of $\bar{x}$ is limited by an integer $s \ll n$ \cite{mallat2008wavelet}. 
Similarly, \emph{block-sparsity} assumes that there is a collection of non-overlapping groups (i.e. blocks) that cover all entries of $\bar{x}$, and that when the unknown signal is restricted to these groups, only a small number of such restrictions are nonzero \cite{YuLi:06}. Further examples include \emph{group-sparsity} \cite{ObLaVe:11}, a generalization of block-sparsity that removes the assumption of non-overlapping groups, {tree-based sparsity \cite{HeCa:09a, HZM:11}, model-based compressed sensing \cite{BaCeDuHe:10}, and sparsity in levels \cite{AdHaPoRo:17}}. 

When recovering signals with {some notion of} sparsity, the typical approach is to construct a regularizer that will guarantee the recovery of an unknown sparse signal with high probability. However, in some applications of practical importance where the unknown signal is {sparse in a non-regular sense}, we still see that $\ell_1$-recovery is the standard recovery method \cite{PaSuEl:17} or that $\ell_1$-recovery performs as good as  recovery via a tailored regularizer \cite{Elhamifar:TSP12}. There are several reasons why recovery via $\ell_1$-minimization {might be} preferable over recovery via regularizers tailored for the sparsity pattern in hand {for some applications}. Probably the most important reason is that $\ell_1$-minimization can be cast as a linear program, hence a solution can be found efficiently even for  large scale problems \cite{YouRoVi:16}. This motivates the exploration of the exact recovery capabilities of $\ell_1$-minimization. In other words, we would like to understand how far the exact recovery capability of $\ell_1$-minimization goes beyond regular sparsity.

In general, a {\emph{sparsity pattern}} can be defined as a collection of index sets. However, some of the most common sparsity patterns studied in the literature satisfy a noteworthy property, which for regular sparsity can be stated as follows: If a vector is $s$-sparse with support $S$, then any vector with support $S^{\prime}\subseteq S$ is also $s$-sparse. More generally, we observe that most common sparsity patterns have the following property:
\begin{property}
\label{property:p} If $S$ belongs to a sparsity pattern, then $S^{\prime}\subseteq S$ also belongs to the same sparsity pattern.
\end{property}
Our claim in this paper is that when the $\ell_1$ regularizer is used, property (P) can serve as the basis for a unified treatment of {sparsity patterns} including those discussed above. 

Understanding the exact recovery capabilities of $\ell_1$-minimization has great importance for dictionaries arising in various applications. In this paper we focus on two examples, namely graph incidence matrices and the partial Discrete Fourier Transform (DFT). Our interest in graph incidence matrices stems from their use as a fundamental representation of graphs, and thus as a natural choice when analyzing network flows. In addition, the detection of sparse structural network changes via observations at the nodes can be modeled as an $\ell_1$-recovery problem, 
where the incidence matrix serves as the dictionary. For instance, this approach is used in \cite{SoYaZu:2015} to detect physical and cyber attacks on power grids. On the other hand, the DFT is one of the most important discrete transforms that influences  applications ranging from image processing to solving partial differential equations \cite{smith:97}.

In this paper we answer the following question: {Given a dictionary, what is the largest sparsity pattern that can be recovered by $\ell_1$-minimization?}
Our specific contributions can be summarized as follows:
\begin{enumerate}
    \item We generalize the well-known Nullspace Property to a family of sparsity patterns described by a mathematical object called an \emph{abstract simplicial complex (ASC)} (\S~\ref{sec:GNUP}). This leads to a characterization of all {sparsity patterns} recoverable via $\ell_1$-minimization (\S~\ref{sec:maxASC}), including the largest one, which we call \emph{maximum abstract simplicial complex (MASC)}.

    \item We provide two characterizations of the MASC associated with a dictionary. The first uses the extreme points of the convex set formed by the intersection of the nullspace of the dictionary and the $\ell_1$-ball (\S~\ref{sec:extp}). The second one is based on the fact that the extreme points and the vectors of minimal support coincide for sets defined as the intersection of a subspace and the $\ell_1$-ball (\S~\ref{sec:minsupp}).
    
    \item For graph incidence matrices, we show that the success of $\ell_1$-minimization is determined by the topology of the graph, {specifically by its simple cycles}. Moreover, we show that the decision of whether all $s$-sparse signals can be recovered via $\ell_1$-minimization can be made in polynomial time, although it is NP-hard in general (\S~\ref{ssec:incmatr}). 
    
    \item {When the dictionary is a partial DFT matrix and the unknown signal is real and its dimension is a prime number, we completely characterize the support sets for which $\ell_1$-recovery is always successful. Under stronger assumptions we show that a computationally more advantageous characterization is possible, and we provide a useful lower bound on the maximum sparsity level for which all signals can be recovered (\S~\ref{ssec:pDFT})}.

    \item We illustrate the importance of our results with experiments on incidence (\S~\ref{ssec:inc_matr_exp}) and  partial DFT matrices (\S~\ref{ssec:pDFT_exp}).    
\end{enumerate}

\section{Preliminaries}
\subsection{Notation}
Given a vector $x\in\mathbb{R}^n$, we denote its $\ell_p$-norm for $p \geq 1$ by $\|x\|_p:= (\sum_{k=0}^{n-1} |x_k|^p)^{1/p}$. 
We denote the unit $\ell_p$-sphere in $\mathbb{R}^n$ by $\mathbb{S}_p^{n-1} := \{x\in \mathbb{R}^n : \|x\|_p=1\}$.
Similarly, the unit $\ell_p$-ball in $\mathbb{R}^n$ is denoted by 
$\mathbb{B}_p^{n} := \{x\in \mathbb{R}^n : \|x\|_p\leq 1\}.$\footnote{In order to avoid confusion, we emphasize once again that throughout this paper, $\mathbb{S}_p^{n-1}$ and $\mathbb{B}_p^{n}$ are objects that consist only of real vectors.}

We denote the function that counts the number of nonzero entries in a vector $x\in \mathbb{R}^n$ by $\|x\|_0$.\footnote{Although $\|\cdot\|_0$ is not a norm, it is common jargon to call it the $\ell_0$-norm.} We say that a vector $x$ is $s$-sparse if $\|x\|_0 \leq s$. 
To emphasize that the vector has precisely $s$ nonzero entries, we say it is \emph{exactly $s$-sparse}.
 
The \emph{codimension} of a $d$-dimensional subspace $\mathcal{V}\subseteq\mathbb{R}^n$ is defined to be $n-d$. The \emph{nullspace} of a matrix $\Phi\in \mathbb{R}^{m\times n}$ will be denoted by $\nullsp(\Phi)$. That is, $\nullsp(\Phi) = \{x\in \mathbb{R}^n : \Phi x = 0\}$. We define $\mathcal{U}_n:= \{0,\dots, n-1\}$. When $S\subseteq \mathcal{U}_n$, we assume that $S$ has the natural ordering, and $S_k$ denotes the $k^{th}$ element of $S$.
For a vector $x\in \mathbb{R}^n$ and index set $S\subseteq \mathcal{U}_n$, we denote the part of $x$ supported on $S$ by $x_S$, so that $x_S \in \mathbb{R}^{|S|}$, where $|S|$ denotes the cardinality of a set $S$.
When we would like to keep the dimension unchanged, we use the projection map $\proj_S : \mathbb{R}^n \to \mathbb{R}^{n}$, which simply projects vectors onto the coordinates indexed by $S$. The complement of $S$ in $\mathcal{U}_n$ is denoted by $S^c$, and the collection of all subsets of $S$ (i.e. the \emph{power set} of $S$) is denoted by $2^{S}$. 

For a nonempty convex set $C\subseteq \mathbb{R}^n$, $\Ext(C)$ denotes the set of extreme points of $C$, which are precisely the points that cannot be written as a nontrivial convex combination of two distinct points in $C$. { The \emph{affine hull} of $C$, denoted by $\aff(C)$, is the smallest affine set in $\mathbb{R}^n$ that contains $C$. Alternatively, it can be characterized as the intersection of all affine sets containing $C$ \cite[p.6]{rock:70}. Note that the affine hull of a point in $\mathbb{R}^n$ consists only of the point itself.} The \emph{relative interior} of $C$ is denoted by $\rinte(C)$. {Formally, it is defined \cite[p.44]{rock:70} as 
$$\rinte(C)\! :=\! \left\{x\!\in\! \aff(C): \exists \varepsilon>0,(x+\varepsilon\mathbb{B}_2^n)\cap \aff(C)\subseteq C)\right\}.$$
Hence, the relative interior of a single point is itself. That is, $\rinte(\{x\})=\{x\}$ for all $x\in \mathbb{R}^n$.}
The closure of $C$ is denoted by $\clo(C)$. {It is formally defined \cite[p.44]{rock:70} as
$$ \clo(C):= \bigcap\left\{C+\varepsilon \mathbb{B}_2^n: \varepsilon>0\right\}.$$
The \emph{relative boundary} of $C$ 
is defined as $\partial C = \clo(C)\setminus\rinte(C)$ \cite[p.44]{rock:70}}. See \cite{rock:70} for additional details.

{Throughout paper we introduce several acronyms. In order to help the reader keep track of them, we include a list here.
\begin{center}
 \begin{tabular}{||c | c ||} 
 \hline
 Acronym & Explanation \\ [0.5ex] 
 \hline\hline
 DFT & Discrete Fourier Transform \\ 
 \hline
 ASC & Abstract Simplicial Complex \\
 \hline
 MASC & Maximum Abstract Simplicial Complex \\
 \hline
 MC  & Mutual Coherence \\
 \hline
 NUP & Nullspace Property \\
 \hline
 RIP & Restricted Isometry Property \\
 \hline
 GNUP & Generalized Nullspace Property \\
 \hline
 MRSL & Maximal  Recoverable  Sparsity  Level \\[1ex] 
 \hline
\end{tabular}
\end{center}
}

\subsection{Sparse Recovery Via $\ell_1$-Minimization}
A na\"ive approach to the recovery of $s$-sparse signals is to pose the $\ell_0$ optimization problem
\begin{equation}\label{eq:l0min}
    \min_{\Phi \bar{x} = \Phi x} \|x\|_0.
\end{equation}
However, it is well-known that \eqref{eq:l0min} is NP-hard to solve. Hence, the following $\ell_1$ convex relaxation is commonly studied
\begin{equation}\label{eq:l1relax}
    \min_{\Phi \bar{x} = \Phi x} \|x\|_1.
\end{equation}
The minimizer of \eqref{eq:l1relax} is unique and coincides with the original signal $\bar{x}$ if the dictionary $\Phi$ satisfies certain properties. Popular properties in the literature include the \emph{Restricted Isometry Property} (RIP) \cite{candes2005decoding,baraniuk2008simple}, \emph{Mutual Coherence} (MC) \cite{donoho2001uncertainty} and \emph{Nullspace Property} (NUP) \cite{cohen2009compressed}. Among these, the RIP and MC are sufficient conditions, whereas the NUP is a necessary and sufficient condition. In this paper we focus on the NUP.

\begin{defn}[Nullspace Property (NUP)]
A matrix $\Phi\in \mathbb{R}^{m\times n}$ satisfies the NUP
of order $s$ if and only if every $\eta\in \nullsp(\Phi)\setminus\{0\}$ and index set $S\subseteq \mathcal{U}_n$ with $|S|\leq s$ satisfy
$$\|\eta_S\|_1 < \|\eta_{S^{c}}\|_1.$$
\end{defn}

The connection between the NUP and exact sparse recovery is captured by the following result.

\begin{thm}\label{thm:NUPthm}{\cite[Thm.~4.5]{FoRa:13}}
Let $\Phi\in \mathbb{R}^{m\times n}$. Any $s$-sparse vector $\bar{x}\in \mathbb{R}^n$ is the unique solution to the optimization problem
\eqref{eq:l1relax}
if and only if the matrix $\Phi$ satisfies the NUP of order $s$. 
\end{thm}
\begin{proof}
{See the proof of \cite[Thm.~4.5]{FoRa:13}.}
\end{proof}

\section{Sparsity Patterns as Abstract Simplicial Complexes}
In this section we introduce the framework in which we study a generalization of the sparse recovery problem. We first introduce a generalization of the NUP, and then show how it naturally leads to the definition of the Maximum Abstract Simplicial Complex (MASC), which encapsulates the sparsity patterns that can be recovered via $\ell_1$-minimization.

\subsection{Generalized Nullspace Property}\label{sec:GNUP}
To generalize Thm.~\ref{thm:NUPthm}, we start by generalizing the NUP. For this purpose, we turn to property (P) and define an Abstract Simplicial Complex (ASC) as any collection of index sets that satisfies a generalization of property (P).

\begin{defn}[Abstract Simplicial Complex (ASC)]
Let $\Omega$ be a nonempty set.  A nonempty set $\mathcal{T}\subseteq 2^{\Omega}$ is called an ASC if and only if for any $S\in \mathcal{T}$ and $W\subseteq S$, we have $W\in \mathcal{T}$.
\end{defn}

It is easy to see that the collection of supports of all $s$-sparse, $s$-block-sparse and $s$-group-sparse signals each form an ASC. A less trivial sparsity pattern that can be associated with an ASC is found in \cite{PaSuEl:17}. There, the authors consider a \emph{convolutional sparse model} $y = Dx$, where the global dictionary $D\in \mathbb{R}^{N\times mN}$ is the concatenation of all shifted versions of a local dictionary $D_L\in \mathbb{R}^{m\times n}$. Due to the special structure of $D$, a predefined collection of patches (i.e. groups) $\{\Lambda_i\}$ on $y$ defines a collection of groups $\{G_i\}$ (with possible overlap) on $x$ through the relation $y=Dx$. Then, for a given signal $x\in \mathbb{R}^n$, they define the norm
$$\|x\|_{0,\infty} = \max_{i} \|x_{G_i}\|_0.$$
The authors are interested in the signals with $\|x\|_{0,\infty} \leq s$. It can be shown that the collection of the supports of such signals forms an ASC. To see this we note that
$$\|x\|_{0,\infty} \leq s \iff |\supprt(x)\cap G_i|\leq s \ \text{for all} \ i,$$
and that  
$\{S\subseteq \mathcal{U}_n: |S\cap G_i|\leq s \ \text{for all} \ i\}$
is an ASC.

We now use the ASC to state our generalization of the NUP.

\begin{defn}[Generalized Nullspace Property (GNUP)]\label{defn:GNUP}
Let $\mathcal{T}\subseteq 2^{\mathcal{U}_n}$. We say that a matrix $\Phi\in \mathbb{R}^{m\times n}$ satisfies the GNUP with respect to (w.r.t.) $\mathcal{T}$
if and only if every $\eta\in \nullsp(\Phi)\setminus\{0\}$ and set $S \in \mathcal{T}$ satisfy
$$\|\eta_S\|_1 < \|\eta_{S^{c}}\|_1.$$
\end{defn}

The GNUP in Defn.~\ref{defn:GNUP} is a \emph{generalization} of the NUP since the latter may be recovered by choosing
\begin{equation}\label{def:Ts}
\mathcal{T} = \mathcal{T}_s:= \{S \in 2^{\mathcal{U}_n}:  |S|\leq s\}
\end{equation}
so that $\mathcal{T}$ is an ASC and the GNUP is equivalent to the NUP.

We are now ready to generalize Thm.~\ref{thm:NUPthm} to the GNUP.

\begin{thm}\label{thm:genthm}
Let $\mathcal{T}\subseteq 2^{\mathcal{U}_n}$ be an ASC. Every $\bar{x}\in \mathbb{R}^n$ with $\supprt(\bar{x})\in \mathcal{T}$ is the unique solution to the optimization problem \eqref{eq:l1relax}
if and only if $\Phi$ satisfies the GNUP 
w.r.t. $\mathcal{T}$.
\end{thm}
\begin{proof}
This theorem is a generalization of the classical Thm.~\ref{thm:NUPthm} whose proof follows mutatis mutandis.

[$\Rightarrow$] Let $\eta\in \nullsp(\Phi)$ be nonzero and $S\in \mathcal{T}$ so that 
$$0=\Phi \eta = \Phi(\proj_S(\eta)-(-\proj_{S^c}(\eta))),$$
which shows using $\bar{x} := \proj_S(\eta)$  that $\Phi \bar{x} = \Phi(-\proj_{S^c}(\eta))$.  Since $\supprt(\bar{x}) \subseteq S\in\mathcal{T}$ and $\mathcal{T}$ is an ASC, it follows that 
$\supprt(\bar{x})\in \mathcal{T}$.
Combining this with the hypothesis for the direction we are proving shows that $\bar{x}$ as the unique solution to \eqref{eq:l1relax}.
Hence, it follows that $\|\eta_S\|_1 = \|\proj_S(\eta)\|_1 = \|\bar{x}\|_1 <  \|-\proj_{S^c}(\eta)\|_1=\|\eta_{S^c}\|_1$, which completes the proof.

[$\Leftarrow$] Let $\bar{x}\in \mathbb{R}^n$ satisfy $S=\supprt(\bar{x}) \in \mathcal{T}$.  Then, let $x\in \mathbb{R}^n$ be any vector satisfying  $\Phi \bar{x} = \Phi x$  and $\bar{x}\neq x$. By setting $\eta = \bar{x}-x$ we see that  $\eta\in \nullsp(\Phi )\setminus\{0\}$ and that
\begin{alignat*}{3}
\|\bar{x}\|_1 &= \|\bar{x}-\proj_S (x) + \proj_S (x)\|_1 && ~\\
        &\leq \|\bar{x}-\proj_S (x)\|_1 + \|\proj_S (x)\|_1 && \quad\text{(triangle inequality)}\\
        &= \|\proj_S (\bar{x}-x)\|_1 + \|x_S\|_1 && \quad\text{(since $S=\supprt(\bar{x})$)}\\
        &= \|\eta_S\|_1 + \|x_S\|_1 && ~\\
        &< \|\eta_{S^c}\|_1 + \|x_S\|_1 && \quad\text{(use the GNUP)}\\
        &= \|-x_{S^c}\|_1 + \|x_S\|_1 && \quad\text{(since $S=\supprt(\bar{x})$)}\\
        &= \|x\|_1 &&\qquad
\end{alignat*}
which proves that $\bar{x}$ is the unique solution to~\eqref{eq:l1relax}.
\end{proof}

A useful alternative formulation of the GNUP is now stated.

\begin{lemma}\label{lem:GNUPreform}
Let $\mathcal{T}\subseteq 2^{\mathcal{U}_n}$. A matrix $\Phi\in \mathbb{R}^{m\times n}$ satisfies the GNUP w.r.t. $\mathcal{T}$ if and only if
\begin{equation}\label{eq:GNUPreform}
    \max_{S\in \mathcal{T}} \max_{\eta\in \nullsp(\Phi)\cap \mathbb{B}_1^n} \|\eta_S\|_1< \tfrac{1}{2}.
\end{equation}
\end{lemma}
\begin{proof}
By definition of the GNUP, the matrix $\Phi$ satisfies the GNUP w.r.t. $\mathcal{T}$ if and only if for all nonzero $\eta\in \nullsp(\Phi)$ and $S\in \mathcal{T}$ it holds that $\|\eta_S\|_1 < \|\eta_{S^c}\|_1$.
Without loss of generality, we can normalize $\eta$, and assume $\|\eta_S\|_1 + \|\eta_{S^c}\|_1 = \|\eta\|_1=1$. Then, $\Phi$ satisfies the GNUP w.r.t. $\mathcal{T}$ if and only if
$$\max_{S\in \mathcal{T}} \max_{\eta\in \nullsp(\Phi)\cap \mathbb{S}_1^{n-1}} \|\eta_S\|_1< \tfrac{1}{2}.$$
Now, since $\mathbb{S}^{n-1}_1$ is the boundary of $\mathbb{B}^n_1$ and including the interior of $\mathbb{B}^n_1$ in the feasible set does not change the solution, we obtain the desired result.
\end{proof}

\begin{rem}\label{rem:nscNP}
When the ASC is given by $\mathcal{T}=\mathcal{T}_s$ (see \eqref{def:Ts}),
i.e. the ASC is the collection of index sets of cardinality less than or equal to $s$, 
the result of the double maximization in \eqref{eq:GNUPreform} is called the nullspace constant, and is NP-hard to compute \cite{TiPf:14} in general. For this special case, we  denote the constant by 
\begin{equation}
\nsc(s, \Phi) := \max_{|S|\leq s} \max_{\eta\in \nullsp(\Phi)\cap \mathbb{B}_1^n} \|\eta_S\|_1.\label{eq:nsc}
\end{equation}
\end{rem}

\subsection{Maximum ASC Associated with a Matrix}\label{sec:maxASC}

If $\mathcal{T}_1$ and $\mathcal{T}_2$ are ASCs with $\mathcal{T}_1, \mathcal{T}_2 \subseteq 2^{\mathcal{U}_n}$ and $\Phi \in \mathbb{R}^{m\times n}$ satisfies the GNUP w.r.t. $\mathcal{T}_1$ and $\mathcal{T}_2$, then it follows from the definition of the GNUP that $\Phi$ also satisfies the GNUP w.r.t. $\mathcal{T}_1\cup \mathcal{T}_2$. This  observation implies that for any matrix $\Phi $, there is a maximum ASC for which $\Phi $ satisfies the GNUP. Since the cardinality of $\mathcal{U}_n$ is  finite, the maximum ASC can be defined as the union of all ASCs for which $\Phi $ satisfies the GNUP.

\begin{defn}[Maximum ASC (MASC)]\label{def:masc}
The union of every ASC in 
$2^{\mathcal{U}_n}$ for which $\Phi$ satisfies the GNUP is called the \emph{MASC associated with $\Phi$}, and is denoted by $\mathcal{T}_{\max}(\Phi)$.
\end{defn}

The next lemma is a consequence of Defn.~\ref{def:masc}.

\begin{lemma}\label{lem:GNUP-wrt-MASC}
Let $\Phi\in\mathbb{R}^{m\times n}$. The  following statements hold:
\begin{enumerate}
    \item[(i)] $\Phi$ satisfies the GNUP w.r.t.  $\mathcal{T}_{\max}(\Phi)$. 
    \item[(ii)] Every $\bar{x}\in \mathbb{R}^n$ with $\supprt(\bar{x})\in \mathcal{T}_{\max}(\Phi)$ is the unique solution to the optimization problem \eqref{eq:l1relax}.\label{lem:GNUP-wrt-MASC/statement:2}
\end{enumerate}
\end{lemma}
\begin{proof}
We first prove part (i).  Let $\eta\in \nullsp(\Phi)\setminus\{0\}$ and $S \in \mathcal{T}_{\max}(\Phi)$.  Then, it holds  from the definition of $\mathcal{T}_{\max}(\Phi)$ that there exists an ASC, say $\mathcal{T}$, that satisfies $S\in \mathcal{T} \subseteq \mathcal{T}_{\max}(\Phi)$ and that $\Phi$ satisfies the GNUP w.r.t. $\mathcal{T}$, which implies, since $S\in\mathcal{T}$, that
$\|\eta_S\|_1 < \|\eta_{S^{c}}\|_1$. This completes the proof.
    
Part (ii) follows by setting $\mathcal{T}=\mathcal{T}_{\max(\Phi)}$ in  Thm.~\ref{thm:genthm}, which is allowed because of  Lem.~\ref{lem:GNUP-wrt-MASC}(i).
\end{proof}

Since $\mathcal{T} = \{\emptyset\}$ is always an ASC, it follows that the MASC  always contains the empty set.  The next example shows that for some matrices, this may be the only set in the MASC.

\begin{exm}
Let $\Phi = \begin{bmatrix} 1 & -1 & 1\end{bmatrix}$ so that the $\nullsp(\Phi)$ is spanned by
$\eta_0:=\begin{bmatrix} 1 & 1 & 0\end{bmatrix}^T$ and $\eta_1:=\begin{bmatrix} 0 & 1 & 1\end{bmatrix}^T$. Hence, any vector $\eta\in \nullsp(\Phi)$ is of the form $\begin{bmatrix} \alpha & \alpha+\beta & \beta\end{bmatrix}$, where $\{\alpha,\beta\}\subset\mathbb{R}$. Then, it may be shown that for each nonempty set $S\subseteq \mathcal{U}_3$, there exists a nonzero $\eta\in \nullsp{(\Phi)}$ such that $\|\eta_S\|_1 \geq \|\eta_{S^c}\|_1 $.  This shows that if $\Phi$ satisfies the GNUP w.r.t. some ASC $\mathcal{T}$, then $\mathcal{T} = \{\emptyset\}$. It follows that $\mathcal{T}_{\max}(\Phi) = \{\emptyset\}$. 
\end{exm}

Lem.~\ref{lem:GNUP-wrt-MASC}(ii) says that any $\bar{x}\in \mathbb{R}^n$ with $\supprt(\bar{x})\in \mathcal{T}_{\max}(\Phi) $  is the unique solution to the optimization problem \eqref{eq:l1relax}. It is natural to ask whether the converse is true. That is, if $\bar{x}$ is the unique solution to the optimization problem \eqref{eq:l1relax}, then is it true that $\supprt(\bar{x}) \in \mathcal{T}_{\max}(\Phi)$? This converse statement is not necessarily true, and we direct the reader to {\cite[Cor.~4.29 \& Thm. 4.30]{FoRa:13}} for additional details on why it may fail.

Although some well-known sparsity patterns form ASCs, it is not true that every useful sparsity pattern is an ASC (e.g., see~\cite{HZM:11}). 
The next result shows that, for \emph{any} $\mathcal{T}\subseteq\mathcal{T}_{\max}(\Phi)$ ($\mathcal{T}$ need not be an ASC), each vector $\bar x$  whose support is in $\mathcal{T}$ is the unique solution to~\eqref{eq:l1relax}. This highlights the fact that the MASC is a special ASC.

\begin{thm}\label{thm:genthm2}
Let $\mathcal{T}\subseteq 2^{\mathcal{U}_n}$ and $\Phi\in\mathbb{R}^{m\times n}$.
Every $\bar{x}\in \mathbb{R}^n$ with $\supprt(\bar{x})\in \mathcal{T}$ is the unique solution to the optimization problem \eqref{eq:l1relax}
if and only if $\mathcal{T}\subseteq \mathcal{T}_{\max}(\Phi)$.
\end{thm}
\begin{proof}
We prove both directions of the implication in turn. 

[$\Leftarrow$]
Let $\mathcal{T}\subseteq \mathcal{T}_{\max}(\Phi)$ and $\bar x$ be so that $\supprt(\bar{x})\in\mathcal{T}\subseteq \mathcal{T}_{\max}(\Phi)$. The result now follows from Lem.~\ref{lem:GNUP-wrt-MASC}(ii).
    
[$\Rightarrow$] For a proof by contradiction,
let it hold that every $\bar{x}\in \mathbb{R}^n$ with $\supprt(\bar{x})\in \mathcal{T}$ is the unique solution to the optimization problem \eqref{eq:l1relax}, but yet
$\mathcal{T}\nsubseteq \mathcal{T}_{\max}(\Phi)$. 
Then, there exists $S\in \mathcal{T}$ such that $S\notin \mathcal{T}_{\max}(\Phi)$. The power set of $S$, namely $2^S$, is the smallest ASC that contains $S$. Since,  $S\notin \mathcal{T}_{\max}(\Phi)$, it holds that $2^S \nsubseteq \mathcal{T}_{\max}(\Phi)$. This means that there exists a nonzero $\eta\in \nullsp(\Phi)$ and $W\in 2^S$ satisfying $\|\eta_{W^c}\|_1 \leq \|\eta_{W}\|_1$, which implies that
\begin{equation}\label{eq:bad_ineq}
    \|\eta_{S^c}\|_1 \leq \|\eta_{S}\|_1.
\end{equation}
Using $0 = \Phi \eta = \Phi (\proj_S(\eta) + \proj_{S^c}(\eta))$, it follows that $\Phi \proj_S(\eta) = \Phi (-\proj_{S^c}(\eta))$. With $\bar \eta = \proj_S(\eta)$ and $\hat \eta = -\proj_{S^c}(\eta)$, it holds using~\eqref{eq:bad_ineq} that $\Phi\bar \eta = \Phi \hat \eta$ and $\|\hat \eta\|_1 \leq \|\bar \eta \|_1$. 
This  contradicts the above uniqueness of solutions if $\supprt(\bar \eta) = S$ since $S\in\mathcal{T}$.
Thus, for the remainder of the proof, we only consider the case $\supprt(\bar \eta) \subsetneqq S$.

Let $\tilde{\eta}\in \mathbb{R}^n$ be any vector satisfying
    (i) $\supprt(\tilde{\eta}) = S$,
    (ii) $\tilde{\eta}_k = \bar \eta_k$ if $k\in \supprt(\bar \eta)$, and
    (iii) $\tilde{\eta}_k$ is any arbitrary nonzero number for all $k\in S\setminus \supprt(\bar \eta)$.
Therefore, it holds that $\|\tilde{\eta} + \bar \eta\|_1 = \|\tilde{\eta} \|_1 + \| \bar \eta\|_1$ and $\supprt(\tilde{\eta} + \bar \eta) = S$. On the other hand, since $\supprt(\tilde{\eta})\cap \supprt(\hat\eta)= \emptyset$, we have $\|\tilde{\eta} + \hat \eta\|_1 = \|\tilde{\eta} \|_1 + \| \hat \eta\|_1$ as well. Note that $\bar{x}:=\tilde{\eta} + \bar \eta \neq \tilde{\eta} + \hat \eta$ because otherwise $\eta \equiv \bar \eta - \hat \eta =0$, which would be a contradiction. Combining these observations shows that
\begin{alignat*}{3}
&    \|\tilde{\eta} + \hat\eta\|_1  &&=  \|\tilde{\eta} \|_1 + \| \hat\eta\|_1 &&~\\
& ~                                  &&\leq \|\tilde{\eta} \|_1 + \| \bar\eta\|_1 && \quad\text{(using \eqref{eq:bad_ineq})}\\
&  ~                                 &&= \|\tilde{\eta} + \bar \eta\|_1.                 && ~
\end{alignat*}
It now follows that
$\Phi\bar{x} = \Phi(\tilde{\eta} +\bar\eta) = \Phi(\tilde{\eta} +\hat\eta)$ with $\supprt(\bar{x}) = S$, but yet $\|\tilde{\eta} +\hat \eta\|_1\leq \|\bar{x}\|_1$, which contradicts the uniqueness of the solutions since $S\in\mathcal{T}$.
\end{proof}

{As a consequence of Thm.~\ref{thm:genthm2}, we see that $\mathcal{T}_{\max}(\Phi)$ is the collection of all support sets $S$ for which $\ell_1$-minimization is always successful. In particular, it can be characterized as
$$\mathcal{T}_{\max}(\Phi) = \{S\subseteq \mathcal{U}_n: \|\eta_S\|_1 < \|\eta_{S^c}\|_1, \forall \eta\in \nullsp(\Phi)\setminus\{0\}\}.$$}
Finally, we note that the MASC  
can be used to obtain a lower bound for the probability of exact recovery.

\begin{prop}\label{prop:probbound}
Let $\mathcal{T}\subseteq 2^{\mathcal{U}_n}$ and  $\Phi\in\mathbb{R}^{m\times n}$.  Assume $\bar{x}\in\mathbb{R}^n$ is randomly chosen from some distribution. Let
$$
p
:= \mathbb{P}(\text{$\bar{x}$ is the unique solution to~\eqref{eq:l1relax}} \mid \supprt(\bar{x})\in \mathcal{T})
$$
and
$$
q:=\mathbb{P}(\supprt(\bar{x})\in \mathcal{T}_{\max}(\Phi) \mid \supprt(\bar{x})\in \mathcal{T}),
$$
Then, it follows that $p \geq q$.
\end{prop}
\begin{proof}
Let $\supprt(\bar{x})\in \mathcal{T}$.  If  $\supprt(\bar{x})\in \mathcal{T}_{\max}(\Phi)$, then it follows from Lem.~\ref{lem:GNUP-wrt-MASC}(ii) 
that $\bar x$ is the unique solution to~\eqref{eq:l1relax}.  It follows from this fact that $p \geq q$, which proves the result.
\end{proof}

The significance of Prop.~\ref{prop:probbound} is that when  $\mathcal{T}$ and $\mathcal{T}_{\max}(\Phi)$ are known, it may be possible to efficiently compute 
$q$ even when $p$ cannot be efficiently calculated. We provide examples later when we consider special matrices in \S~\ref{ssec:inc_matr_exp} and \S~\ref{ssec:pDFT_exp}.

\section{Characterizations of the GNUP and the MASC}
\label{sec:char-gnup-and-masc} 
The definition of the GNUP involves a condition on \emph{all} nonzero vectors in the nullspace. Since this definition is complicated to use in practice, we seek alternative characterizations that provide new insights and computational advantages, which translate over to the MASC. In this section we provide two such characterizations: the first is based on extreme points and the second is based on vectors of minimal support. An advantage of the characterization based on extreme points is that it is geometrically intuitive, whereas the characterization based on the vectors of minimal support leverages existing results in the literature better. This latter point will become clear when we discuss special classes of matrices in \S~\ref{ssec:inc_matr_exp} and \S~\ref{ssec:pDFT_exp}.

\subsection{Characterization in Terms of Extreme Points}\label{sec:extp}
In this section we reformulate the GNUP and MASC in terms of the extreme points of $\nullsp(\Phi)\cap \mathbb{B}_1^n$.

\begin{lemma}\label{lem:GNUPreformExt}
Let $\mathcal{T}\subseteq 2^{\mathcal{U}_n}$. A matrix $\Phi\in \mathbb{R}^{m\times n}$ satisfies the GNUP w.r.t. $\mathcal{T}$ if and only if  $$\max_{S\in \mathcal{T}} \max_{z\in \Ext(\nullsp(\Phi)\cap \mathbb{B}_1^n)} \|z_S\|_1< \tfrac{1}{2}.$$
\end{lemma}

\begin{proof}
For each $S\in \mathcal{T}$, the problem
$\max_{z\in \nullsp(\Phi)\cap \mathbb{B}_1^n} \|z_s\|_1$ in~\eqref{eq:GNUPreform} 
is the maximization of a convex function
over the nonempty, compact, and convex set $\nullsp(\Phi)\cap \mathbb{B}_1^n$. Hence, the maximum is attained at an extreme point \cite{rock:70}. Thus, it is enough to maximize only over $\Ext(\nullsp(\Phi)\cap \mathbb{B}_1^n)$, which with Lem.~\ref{lem:GNUPreform} gives the result.
\end{proof}

Lem.~\ref{lem:GNUPreformExt} allows for a clearer description of 
the MASC associated with a matrix $\Phi$ in terms of $\Ext(\nullsp(\Phi) \cap \mathbb{B}_1^n$).

\begin{prop}\label{prop:maxASC}
For any matrix $\Phi\in \mathbb{R}^{m\times n}$, we have
\begin{align*}
\mathcal{T}_{\max}(\Phi) = \{S\subseteq \mathcal{U}_n : 
    &\|z_S\|_1< \tfrac{1}{2} \ 
    \forall z \in\Ext(\nullsp(\Phi) \cap \mathbb{B}_1^n)\}. 
\end{align*}
\end{prop}

\begin{proof}
We prove that each set is contained in the other.

[$\subseteq $] 
It follows from Lem.~\ref{lem:GNUP-wrt-MASC}(i) and Lem.~\ref{lem:GNUPreformExt} that
$\|z_S\|< 1/2$ for all $z\in \Ext(\nullsp(\Phi)\cap \mathbb{B}_1)$ and all $S\in \mathcal{T}_{\max}(\Phi)$, which establishes that this inclusion holds.

[$\supseteq$] It is easy to see that the right-hand side is an ASC, and then Lem.~\ref{lem:GNUPreformExt} shows that $\Phi$ satisfies the GNUP w.r.t. the right-hand side.
Therefore, by the definition of $\mathcal{T}_{\max}(\Phi)$, it must hold that the right-hand side is contained in $\mathcal{T}_{\max}(\Phi)$.
\end{proof}

The extreme points of $\nullsp(\Phi)\cap \mathbb{B}^n_1$ characterize the GNUP and MASC associated with a dictionary $\Phi$. Although this characterization has a clear geometric interpretation, the computation of extreme points is non-trivial.  Therefore, in the next section we provide an alternative mathematical characterization based on vectors of minimal support that can be more easily described, especially for certain classes of matrices.

\subsection{Characterization in Terms of Vectors of Minimal Support}\label{sec:minsupp}

We now delve deeper into the properties of the extreme points of $\nullsp(\Phi)\cap \mathbb{B}^n_1$.  In particular, our aim is to associate them with \emph{vectors of minimal support},\footnote{In \cite{rock:70}, vectors of minimal support are called \emph{elementary vectors}.} which we now define.

\begin{defn}[Vectors of minimal support]
Let $\mathcal{V}\subseteq \mathbb{R}^n$ be a subspace and $x\in\mathcal{V}$ be a nonzero vector. We say that $x$ has \emph{minimal support in $\mathcal{V}$} if and only if there does not exist a nonzero $\bar{x}\in \mathcal{V}$ whose support is contained in, but not equal to, the support of $x$, i.e., $\supprt(\bar{x})\subsetneqq \supprt(x)$. We denote the collection of all such vectors in $\mathcal{V}$ by $\MinSupp(\mathcal{V})$.
\end{defn}

To illustrate the concept of vectors of minimal support, let us consider the following example.

\begin{exm}
Let $\mathcal{V}\subset \mathbb{R}^3$ be the subspace defined by 
\begin{equation}\label{eq:subspace}
    x_0+x_1+x_2 =0.
\end{equation}
Let $x\in \MinSupp(\mathcal{V})$. By definition, $x$ is nonzero. Moreover, $\|x\|_0 >1$ since otherwise \eqref{eq:subspace} cannot be satisfied. On the other hand, $\|x\|_0\neq 3$ because there exist $2$-sparse vectors in $\mathcal{V}$ whose supports are obviously contained in $\{0,1,2\}$. 
Therefore, $\MinSupp(\mathcal{V})$ contains exactly $2$-sparse vectors, i.e., 
$$
\MinSupp(\mathcal{V}) = \{x\in \mathbb{R}^3: \|x\|_0=2 \text{ and }x_0+x_1+x_2 =0\}.
$$
\end{exm}

Our next goal is to establish that the sets $\Ext(\mathcal{V}\cap\mathbb{B}_1^n)$ and $\MinSupp(\mathcal{V}) \cap\mathbb{S}_1^{n-1}$ are equal. To prove this result, we must first motivate and prove an auxiliary result. Specifically, let
\begin{equation}
\Delta_{{e}} := \{x\in \mathbb{S}^{n-1}_1: \langle {e}, x\rangle =1\}
\end{equation}
be a (closed and convex) simplex defined by a nonzero \emph{sign vector} ${e}\in \{0,1,-1\}^n\subset \mathbb{R}^n$. Then, $\mathbb{S}_1^{n-1}$ can be written as a disjoint union of the relative interiors of such simplexes, i.e.,
\begin{equation}\label{eq:S1decomp}
\mathbb{S}^{n-1}_1 \equiv \bigcup^{\circ}_{{e}} \rinte(\Delta_{{e}}).
\end{equation}
We can use \eqref{eq:S1decomp} to give a characterization of $\Ext(\mathcal{V}\cap \mathbb{B}_1^n)$.

\begin{lemma}\label{lem:extp}
Let $\mathcal{V}\subseteq \mathbb{R}^n$ be a subspace of dimension $d \geq 1$, and define its codimension as $r := n -d$. The following hold. 
\begin{itemize}
    \item[(i)] $z\in \Ext(\mathcal{V}\cap \mathbb{B}_1^n)$ if and only if $\{z\} = \mathcal{V}\cap \rinte(\Delta_{{e}}) $ for some nonzero sign vector ${e}$.
    \item[(ii)] 
    If $z\in\Ext( \mathcal{V}\cap \mathbb{B}_1^n)$, then $z$ is  $(r+1)$-sparse.
\end{itemize}
\end{lemma}
\begin{proof}
We first prove part (i).  For the only if direction, let $z\in \Ext(\mathcal{V}\cap \mathbb{B}_1^n)$, which means that $z\in \mathcal{V}\cap \mathbb{S}_1^n$ because $d \geq 1$.
Since it follows from  \eqref{eq:S1decomp} that
\begin{equation}\label{eq:VrinteSimp}
    \mathcal{V}\cap \mathbb{S}^{n-1}_1 
    = \bigcup^{\circ}_{{e}} \left(\mathcal{V}\cap \rinte(\Delta_{{e}})\right),
\end{equation}
we know that there exists a unique nonzero sign vector ${e}$ such that $z\in \mathcal{V}\cap \rinte(\Delta_{{e}})$.
If $\dim(\mathcal{V}\cap \rinte(\Delta_{{e}})) > 0$, then any point in $\mathcal{V}\cap \rinte(\Delta_{{e}})$ can be written as a nontrivial convex combination of two distinct points in $\mathcal{V}\cap \rinte(\Delta_{{e}})$, which is a contradiction. 
Therefore, it must be that $\dim(\mathcal{V}\cap \rinte(\Delta_{{e}})) = 0$ and that $\{z\} = \mathcal{V}\cap \rinte(\Delta_{{e}})$, as claimed.

For the converse direction, let $z$ be such that $\{z\} = \mathcal{V}\cap \rinte(\Delta_{{e}})$ for some nonzero sign vector ${e}$. Since  $z$ is the unique maximizer of the linear functional $\langle{e}, \cdot\rangle$ on the compact polyhedral set $\mathcal{V}\cap \mathbb{B}_1^n$, we must conclude that $z\in \Ext(\mathcal{V}\cap \mathbb{B}_1^n)$.

We now prove part (ii).
Let $z\in \Ext(\mathcal{V}\cap\mathbb{B}_1^n)$ so that from part (i) we have $\{z\}= \mathcal{V}\cap \rinte(\Delta_{{e}})$ for some nonzero sign vector ${e}$.
Since $z\in \rinte(\Delta_{{e}})$, we necessarily have
\begin{equation}\label{eq:sparsityExtPt}
    \|z\|_0 = \|{e}\|_0 = \dim(\Delta_{{e}})+1.
\end{equation}
Now, to reach a contradiction, suppose that $\dim(\Delta_{{e}}) > r$. Combining this with the fact that 
$\mathcal{V}\cap \rinte(\Delta_{{e}}) \neq \emptyset$, shows necessarily that $\dim(\mathcal{V}\cap \rinte(\Delta_{{e}})) >0$, which is a contradiction.  We have shown that $\dim(\Delta_{{e}}) \leq r$, which combined with \eqref{eq:sparsityExtPt} gives the desired result.
\end{proof}

We now prove the equivalence between  $\Ext(\mathcal{V}\cap \mathbb{B}_1^n)$ and the intersection of the minimal support vectors of $\mathcal{V}$ and $\mathbb{S}^{n-1}_1$. Fig.~\ref{fig:l1_ball_inters_oblique_plane_3D} illustrates this equivalence for a subspace $\mathcal{V}$ of $\mathbb{R}^3$.

\begin{prop}\label{prop:minSuppExt}
Let $\mathcal{V}\subseteq \mathbb{R}^n$ be a nontrivial subspace. Then, $\MinSupp(\mathcal{V}) \cap \mathbb{S}^{n-1}_1= \Ext(\mathcal{V}\cap \mathbb{B}^{n}_1)$.
\end{prop}
\begin{proof}
We prove that each set is contained in the other.
[$\subseteq$]
Let $z\in \MinSupp(\mathcal{V})\cap \mathbb{S}^{n-1}_1$. 
We know from \eqref{eq:VrinteSimp} that  there exists a unique nonzero sign vector ${e}$ such that $z\in \mathcal{V}\cap \rinte(\Delta_{e})$. If $\dim(\mathcal{V}\cap \rinte(\Delta_{{e}})) = 0$, then it follows from Lem.~\ref{lem:extp}(i) that $z$ is an extreme point of $\mathcal{V}\cap \mathbb{B}^{n}_1$, and we are done. 
Therefore, in the remainder of the proof, we show that $\dim(\mathcal{V}\cap \rinte(\Delta_{{e}})) > 0$ cannot hold. 

For a proof by contradiction, suppose that
$\dim(\mathcal{V}\cap \rinte(\Delta_{{e}})) > 0$.  Define $C$ as the closure of $\mathcal{V}\cap \rinte(\Delta_{{e}})$ in $\mathcal{V}\cap\Delta_{{e}}$ so that the boundary satisfies $\partial C \subseteq\mathcal{V}\cap \mathbb{S}^{n-1}_1$. However, for any point $\bar{z}\in \partial C$ we have $\supprt(\bar{z})\subsetneqq \supprt(z)$, which contradicts the fact that $z$ has minimal support. 

[$\supseteq$] Let $z\in \Ext(\mathcal{V}\cap \mathbb{B}^n_1)$, which means that $z\in \mathcal{V}\cap \mathbb{S}_1^{n-1}$ because $\mathcal{V}$ is a nontrivial subspace by assumption. For a proof by contradiction, suppose that $z \notin \MinSupp(\mathcal{V})\cap\mathbb{S}^{n-1}$, which combined with $z\in\mathcal{\mathcal{V}}\cap\mathbb{S}_1^{n-1}$ implies that $z \notin \MinSupp(\mathcal{V})$, i.e., 
that there exists a nonzero $\bar{z}\in \mathcal{V}\cap\mathbb{S}_1^{n-1}$ 
such that $\supprt(\bar{z})\subsetneqq \supprt(z)$. Next, we define $s:=\|z\|_0$ and $\mathcal{W}\subseteq \mathbb{R}^n$ as the subspace
$$
\mathcal{W}:= \{x\in \mathbb{R}^n: x_k=0\text{ for all }k\notin \supprt(z)\},
$$
and note that $\{z,\bar{z}\}\subseteq \mathcal{W}\cap\mathcal{V}\cap\mathbb{S}^{n-1}_1$ and that $\mathcal{W}$ can trivially be associated with $\mathbb{R}^s$. Under this association, $\mathcal{W}\cap \mathbb{B}^n_1$ corresponds to $\mathbb{B}^{s}_1\subset \mathbb{R}^s$, and the subspace $\mathcal{W}\cap \mathcal{V}$  corresponds to a subspace $\bar{\mathcal{V}}\subset \mathbb{R}^s$. Under this correspondence,
since $z\in\Ext(\mathcal{V}\cap \mathbb{B}^n_1)$ it follows that
$z \in \Ext(\bar{\mathcal{V}}\cap\mathbb{B}_1^s)$.
Since $z$ is $s$-sparse, it follows from Lem.~\ref{lem:extp} that  $s\leq (s-\dim(\bar{\mathcal{V}}))+1$, which shows that $\dim(\bar{\mathcal{V}})\leq 1$. Moreover, since $\bar{\mathcal{V}}$ is nonempty, we necessarily have  $\dim(\bar{\mathcal{V}})=1$ so that  $\supprt(z)= \supprt(\bar{z})$, which is a contradiction. 
\end{proof}

\begin{figure}
\centering
\subfloat[$\mathbb{B}^{3}_1$ and $\mathcal{V}$ in $\mathbb{R}^3$.]{\includegraphics[width=0.21\textwidth]{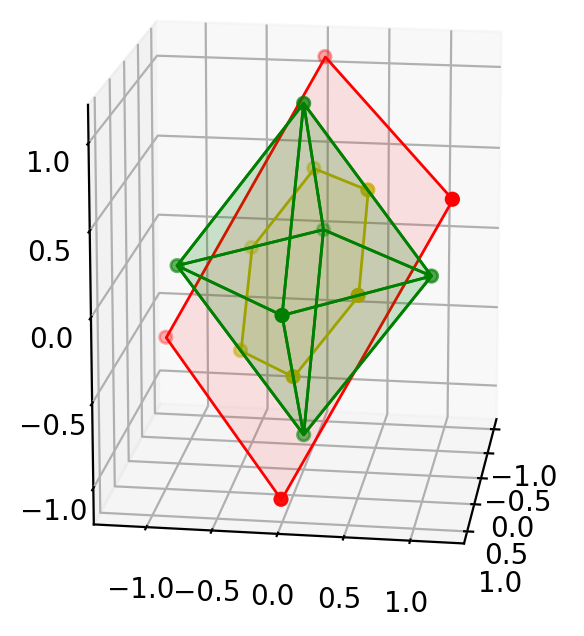}}
~
\subfloat[The intersection of $\mathbb{B}^{3}_1$ and $\mathcal{V}$.]{\includegraphics[width=0.28\textwidth]{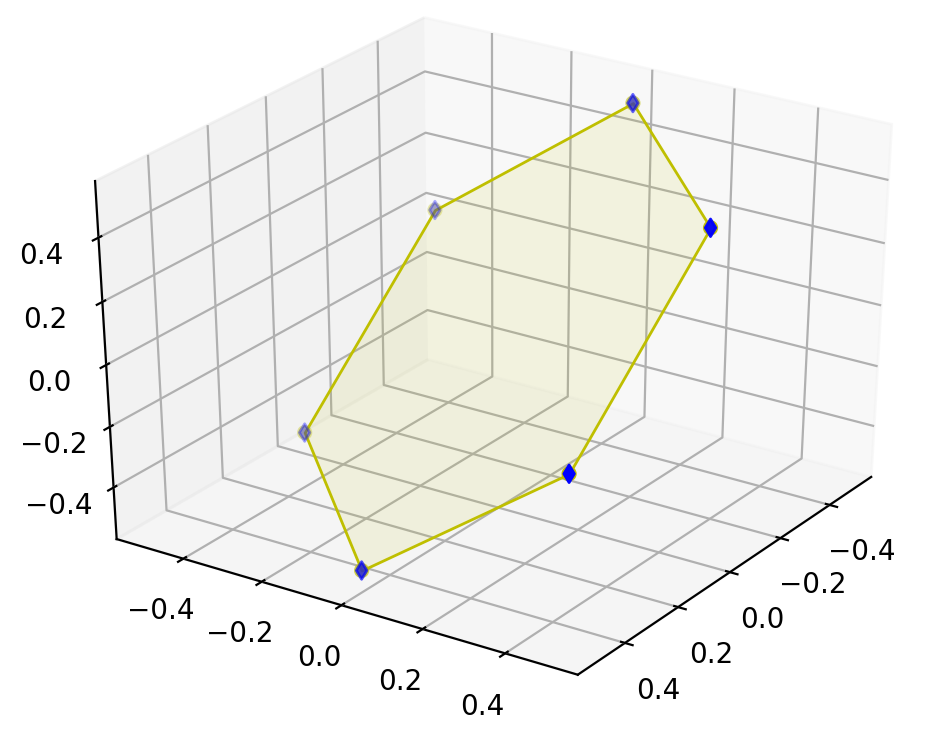}}%
\caption{When a subspace $\mathcal{V}\subset\mathbb{R}^n$ intersects the $\ell_1$-ball $\mathbb{B}_1^n$, the sets $\Ext(\mathcal{V}\cap \mathbb{B}^n_1)$ and $\MinSupp(\mathcal{V})\cap \mathbb{S}^{n-1}_1$ coincide.}
\label{fig:l1_ball_inters_oblique_plane_3D}
\end{figure}

Prop.~\ref{prop:minSuppExt} can be used to characterize the GNUP and MASC associated with a matrix in terms of vectors of minimal support, as stated in the next lemma.
\begin{lemma}\label{lem:GNUPreformminSupp}
Let $\mathcal{T}\subseteq 2^{\mathcal{U}_n}$ and $\Phi\in \mathbb{R}^{m\times n}$ have a nontrival nullspace. Then, $\Phi$ satisfies the GNUP w.r.t. $\mathcal{T}$ if and only if  $$\max_{S\in \mathcal{T}} \max_{z\in \MinSupp(\nullsp(\Phi)) \cap \mathbb{S}_1^{n-1}} \|z_S\|_1< \tfrac{1}{2}.$$
\end{lemma}
\begin{proof}
The result follows from Lem.~\ref{lem:GNUPreformExt} and Prop.~\ref{prop:minSuppExt} using the identification  $\mathcal{V} = \nullsp(\Phi)$ in Prop.~\ref{prop:minSuppExt}.
\end{proof}

\begin{prop}\label{prop:maxASCminSupp}
If $\Phi\in \mathbb{R}^{m\times n}$ has a nontrivial nullspace, then 
\begin{align*}
&\mathcal{T}_{\max}(\Phi) \\
&= \{S\subseteq \mathcal{U}_n : 
    \|z_S\|_1< \tfrac{1}{2} \ 
    \forall z \in\MinSupp(\nullsp(\Phi)) \cap \mathbb{S}_1^{n-1}\}.
\end{align*}
\end{prop}
\begin{proof}
The result follows from Prop.~\ref{prop:maxASC} and Prop.~\ref{prop:minSuppExt} using the identification $\mathcal{V} = \nullsp(\Phi)$ in Prop.~\ref{prop:minSuppExt}.
\end{proof}

We will see in the next two sections that for certain popular dictionaries $\Phi$, the set $\MinSupp(\nullsp(\Phi))$ has a rather simple characterization. Hence, Prop.~\ref{prop:maxASCminSupp} will provide an efficient way to characterize $\mathcal{T}_{\max}(\Phi)$ for these dictionaries.

\section{Application to Graph Incidence Matrices}
The first application we consider is the $\ell_1$-minimization problem with a dictionary that takes the form of an incidence matrix $A$ associated with a simple connected graph $\mathcal{G}$, namely
\begin{equation}\label{eq:incl1min}
    \min_{A\bar{x} = Ax} \|x\|_1.
\end{equation}
We want to understand the MASC associated with $A$.

\subsection{Review of Graph Theory}
We begin with a review of some concepts from graph theory. (See \cite{GoRo:01} for additional details.)
A directed graph with vertex set $V=\{v_1, \dots, v_m\}$ and edge set $E=\{e_1,\dots, e_n\}\subseteq V\times V$ will be denoted by $\mathcal{G} = \mathcal{G}(V,E)$. Hence, each $e_k$ is represented by an ordered pair of vertices $(v_{k_1}, v_{k_2})$.  A graph $\mathcal{G}$ is called \emph{simple} if there is at most one edge connecting any two vertices and no edge starts and ends at the same vertex, i.e. $\mathcal{G}$ has no self loops. Henceforth, $\mathcal{G}$ will always denote a simple directed graph with a finite number of edges and vertices. Associated with a directed graph $\mathcal{G} = \mathcal{G}(V,E)$, we define the incidence matrix $A=A(\mathcal{G}) \in \mathbb{R}^{m \times n}$ as
\begin{equation}
a_{ij}=
\begin{cases}
-1, &\text{ if } v_i \text{ is the initial vertex of edge } e_j, \\
\phantom{-}1, &\text{ if } v_i \text{ is the terminal vertex of edge } e_j,\\
\phantom{-}0, &\text{ otherwise.}
\end{cases}
\end{equation}
The subspace $\nullsp(A)$ is called the \emph{flow space} of $\mathcal{G}$. 
A \emph{simple cycle} is a closed walk with no repetitions of vertices other than the starting and ending vertex, such as $\mathcal{C} = (u_1,\dots, u_r, u_{r+1}=u_1)$.
The length of the shortest simple cycle in $\mathcal{G}$ is called the \emph{girth} of $\mathcal{G}$.
A simple cycle $\mathcal{C} = (u_1,\dots, u_r, u_{r+1}=u_1)$ can be associated with a vector $w(\mathcal{C})\in \{-1,0,1\}^{|E|}$, where each coordinate 
is defined as
\begin{equation*}\label{def.w}
[w(\mathcal{C})]_j =
\begin{cases}
\phantom{-}1, &\text{ if } e_j = (u_i,u_{i+1})\text{ for some }i\in \{1,\dots, r\},\\
-1, &\text{ if } e_j= (u_{i+1}, u_{i})\text{ for some }i\in \{1,\dots, r\},\\
\phantom{-}0, &\text{ otherwise.}
\end{cases}
\end{equation*}
The vector $w(\mathcal{C})$ is called the \emph{signed characteristic vector} of the oriented cycle $\mathcal{C}$ \cite{GoRo:01}. These vectors play a central role.

\begin{rem}\label{rem:minSupport}
The nonzero elements of the flow
space with minimal support are nonzero scalar multiples of the signed characteristic
vectors of the simple cycles of $\mathcal{G}$ \cite[Thm.~14.2.2]{GoRo:01}.
\end{rem}

The signed characteristic vectors of the simple cycles of a simple graph $\mathcal{G}$ span its flow space \cite[Cor.~14.2.3]{GoRo:01}.
We denote the set of normalized signed characteristic vectors of the simple cycles of a simple directed graph $\mathcal{G}$ by 
\begin{equation}\label{eq:simplecycles}
W_1(\mathcal{G}) \! :=\! \left\{\frac{w(\mathcal{C})}{\| w(\mathcal{C})\|_1} :\;\mathcal{C} \text{ is a simple cycle of }\mathcal{G}\right\}.
\end{equation}
Rem.~\ref{rem:minSupport} states, using our notation, that
\begin{equation}\label{eq:W1MinSupp}
    W_1(\mathcal{G}) = \MinSupp(\nullsp(A))\cap \mathbb{S}^{n-1}_1,
\end{equation}
where $A\in\mathbb{R}^{m\times n}$ is the incidence matrix of the graph $\mathcal{G}$.

\subsection{Characterizing the MASC for Graph Incidence Matrices}\label{ssec:incmatr}

We apply the results of Section~\ref{sec:char-gnup-and-masc} to derive necessary and sufficient conditions for exact recovery.
The next corollary provides the link between $\Ext(\nullsp(A)\cap \mathbb{B}_1^n)$ and the simple cycles of the graph $\mathcal{G}${, thus allowing a clear characterization $\mathcal{T}_{\max}(A)$. This is the main result of this section}. 

\begin{cor}\label{cor:extremepts}
Let $\mathcal{G}$ be a simple graph with incidence matrix $A\in \mathbb{R}^{m \times n}$ whose nullspace is nontrivial. The following hold:
\begin{itemize}
    \item[(i)] $\Ext(\nullsp(A)\cap \mathbb{B}_1^{n}) = W_1(\mathcal{G})$.
    \item[(ii)] The MASC is given by
    \begin{align*}
\!\!\!\!\!\!\!\!\mathcal{T}_{\max}(A) 
= \left\{ S\subseteq \mathcal{U}_n :  \frac{|S\cap \supprt(z)|}{\|z\|_0} < \tfrac{1}{2} 
    \ \forall z \in W_1(\mathcal{G}) \right\}. 
\end{align*}
\end{itemize}
\end{cor}
\begin{proof}
Part (i) follows from Prop.~\ref{prop:minSuppExt} and \eqref{eq:W1MinSupp}. To prove part (ii), observe  
for any index set $S \subseteq \mathcal{U}_n$ and $z\in W_1(\mathcal{G})$ that
\begin{equation}\label{eq:zs}
\|z_S \|_1 = \frac{|S\cap \supprt(z)|}{\|z\|_0}.
\end{equation}
Combining this with Prop.~\ref{prop:maxASCminSupp} and \eqref{eq:W1MinSupp} gives the result.
\end{proof}

Combining Cor.~\ref{cor:extremepts}(ii) with Thm.~\ref{thm:genthm2} tells us that the following holds when the dictionary is the incidence matrix of a simple graph: every signal whose support is contained in an index set $S\subseteq \mathcal{U}_n$ can be recovered via $\ell_1$-minimization if and only if for each simple cycle, the size of the intersection of $S$ with the indices of the edges appearing in the simple cycle is strictly less than half of the length of the simple cycle. {Next, as a concrete example, we calculate MASC for the graph in Fig.~\ref{fig:inc_matr_exm}}

\begin{figure}
\subfloat[Graph $\widehat{\mathcal{G}}$.]{\!\!
\begin{minipage}{0.2\textwidth}
\[
\xymatrix@C=2em{
\xy*{1}*\cir<6pt>{}\endxy \ar[r]^{a} \ar[dr]_c &\xy*{2}*\cir<6pt>{}\endxy \ar[d]^b & \ar[l]_g\xy*{6}*\cir<6pt>{}\endxy \\
    &\xy*{3}*\cir<6pt>{}\endxy  \ar[d]_d \\
    &\xy*{4}*\cir<6pt>{}\endxy \ar[r]_{e} &\xy*{5}*\cir<6pt>{}\endxy \ar[uu]_{f}
}
\]
\end{minipage}
}
\setlength{\arraycolsep}{3pt}\!\!
\subfloat[Incidence matrix $A(\widehat{\mathcal{G}})$ of $\widehat{\mathcal{G}}$.]{
\begin{minipage}{0.3\textwidth}
\[
\begin{bmatrix}
-1 & 0 & -1 & 0 &  0 &  0 &  0 \\[1.5ex]
 1 & -1 & 0 & 0 &  0 &  0 &  1 \\[1.5ex]
0 & 1 & 1 & -1 &  0  &  0 &  0 \\[1.5ex]
0 & 0 & 0 & 1 &  -1  &  0 &  0 \\[1.5ex]
0 & 0 & 0 & 0 &  1  &  -1 &  0 \\[1.5ex]
0 & 0 & 0 & 0 &  0  &  1 &  -1 \\
\end{bmatrix}
\]
\end{minipage}
}%
\caption{{The graph $\widehat{\mathcal{G}}$ of Exm.~\ref{exm:inc_matr_exm} and its incidence matrix $A(\widehat{\mathcal{G}})$.}}
\label{fig:inc_matr_exm}
\end{figure}
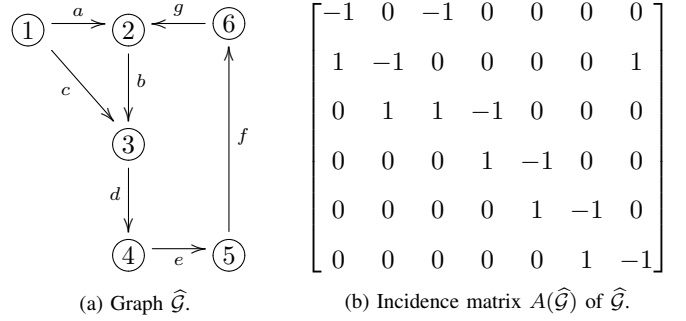
{
\begin{exm}\label{exm:inc_matr_exm}
Consider the simple directed graph $\widehat{\mathcal{G}}$ and its incidence matrix, $A(\widehat{\mathcal{G}})$, as shown in Fig.~\ref{fig:inc_matr_exm}. We denote the edges $E(\widehat{\mathcal{G}})$ of $\widehat{\mathcal{G}}$ with small letters, whereas the vertices of $\widehat{\mathcal{G}}$ are denoted by numbers. The columns of the incidence matrix $A(\widehat{\mathcal{G}})$ are indexed by the edges of $\widehat{\mathcal{G}}$ in alphabetical order. That is, the first column corresponds to the edge ``a", the second column corresponds to edge ``b", etc. The rows of $A(\widehat{\mathcal{G}})$ are indexed by the vertices of $\widehat{\mathcal{G}}$ in their natural ordering. Note that $\widehat{\mathcal{G}}$ has 3 simple cycles, and---up to a sign change---they correspond to the following vectors:
\begin{align*}
    v_1 &:= \begin{bmatrix}1 & 1 & -1 & 0 & 0 & 0 & 0\end{bmatrix}\\
    v_2 &:= \begin{bmatrix}0 & 1 & 0 & 1 & 1 & 1 & 1\end{bmatrix}\\
    v_3 &:= \begin{bmatrix}-1 & 0 & 1 & 1 & 1 & 1 & 1\end{bmatrix}.
\end{align*}
As a consequence of Cor.~\ref{cor:extremepts}(ii), we know that any signal $\bar{x}$ with $\supprt(\bar{x})= \{a,e\}$ can be recovered via \eqref{eq:incl1min}, because $\{a,e\}$ contains strictly less than half of the edges in any simple cycle of $\widehat{\mathcal{G}}$. However, if we have $\supprt(\bar{x})= \{a,b\}$, then \eqref{eq:incl1min} may not be able to recover $\bar{x}$, because the simple cycle that corresponds to $v_1$ consists of the edges $\{a,b,c\}$ and this $\supprt(\bar{x})$ covers $2/3$ of this simple cycle. In fact, Cor.~\ref{cor:extremepts}(ii) tells us that there are signals $\bar{x}$ with $\supprt(\bar{x})= \{a,b\}$, which cannot be recovered via \eqref{eq:incl1min}. Therefore, when the measurement matrix is $A(\widehat{\mathcal{G}})$, the recovery performance of \eqref{eq:incl1min} is not uniform on $2$-sparse signals. Nonetheless, any $1$-sparse signal can always be recovered. A slightly more detailed analysis reveals that there is always some $s$-sparse signal that cannot be recovered, whenever $s\geq 3$. Hence, MASC of $A(\widehat{\mathcal{G}})$ is
\begin{align*}
\mathcal{T}_{\max}(A(\widehat{\mathcal{G}}))\! =\! \{S\in 2^{E(\widehat{\mathcal{G}})}: \, &|S|\leq 1\text{ or }\\
                                                                    &(|S|\!=\!2\text{ and }|S\cap \{a,b,c\}|\leq 1) \}.    
\end{align*}
\end{exm}
}

\begin{rem}
Cor.~\ref{cor:extremepts}(ii) is especially relevant when studying shortest paths in unweighted graphs. Suppose that we are given a simple, undirected, unweighted graph $\mathcal{G}$, for which the shortest path between any pair of vertices is unique. Such graphs are called \emph{geodetic graphs} in the literature \cite{Ore:62}. A tree is an example of a geodetic graph. Moreover, a geodetic graph cannot contain a simple cycle of even length. The crucial observation from our view point is that in a geodetic graph, shortest paths can only contain strictly less than half of the edges in any simple cycle of the graph. Hence, if we orient our graph $\mathcal{G}$ by picking directions for the edges, and denote the oriented graph by $\tilde{\mathcal{G}}$, then by Cor.~\ref{cor:extremepts}(ii), the support of the shortest path between any two vertices of $\tilde{\mathcal{G}}$ is in $\mathcal{T}_{\max}(A(\tilde{\mathcal{G}}))$. Thus, it can be recovered via the $\ell_1$-minimization problem.
\end{rem}

Cor.~\ref{cor:extremepts}(ii) can be specialized to derive necessary and sufficient conditions for the recovery of $s$-sparse signals.

\begin{thm}\label{thm:girthsparsity}
If $A \in \mathbb{R}^{m\times n}$ is the incidence matrix of a simple graph $\mathcal{G}$ with girth $g$, then the nullspace constant \eqref{eq:nsc} satisfies 
$$\nsc(s,A) = \min\left\{1,s/g\right\}.$$
Hence, every $s$-sparse vector $\bar{x}\in \mathbb{R}^n$  is the unique solution to the optimization problem \eqref{eq:incl1min}
if and only if $s < g/2$.
\end{thm}
\begin{proof}
It follows from  the definition of the nullspace constant in~\eqref{eq:nsc}, Cor.~\ref{cor:extremepts}, and \eqref{eq:zs} that
\begin{equation} \label{nsc-A}
\nsc(s,A) = \max_{|S|\leq s} \max_{z\in W_1(\mathcal{G})} \frac{|S\cap \supprt(z)|}{\|z\|_0}.
\end{equation}
Now, consider any $z\in W_1(\mathcal{G})$. It follows that there exists a simple cycle $\mathcal{C}$ of $\mathcal{G}$ with length $g + \ell(z)$ for some $\ell(z) \geq 0$ such that  $z = w(\mathcal{C})/\|w(\mathcal{C})\|_1$, where we have used the fact that the girth of $\mathcal{G}$ is $g$; note that $\|z\|_0 = g+\ell(z)$. For this given $z$, the right-hand side of~\eqref{nsc-A} is maximized by any $S(z)$ satisfying that $|S(z)|\leq s$ and it contains as many indices from the support of $z$ as possible. This choice and~\eqref{nsc-A} show that
\begin{align}
\nsc(s,A)
&= \max_{z\in W_1(\mathcal{G})} \frac{|S(z)\cap \supprt(z)|}{\|z\|_0} \nonumber\\
&= \max_{z\in W_1(\mathcal{G})} \frac{\min\{g+\ell(z),s\}}{g+\ell(z)}  \nonumber\\
&= \max_{z\in W_1(\mathcal{G})} \min\left\{ 1,\frac{s}{g+\ell(z)} \right\}, \label{nsc-A-3}
\end{align}
from which it is now clear that the maximizer is obtained by any $z$ satisfying $l(z) = 0$, i.e., by $z = w(\mathcal{C})/\|w(\mathcal{C})\|_1$ for any simple cycle $\mathcal{C}$ whose length is equal to $g$. Combining this with~\eqref{nsc-A-3} proves that $\nsc(s,A) = \min\{1,s/g\}$, as claimed.  

To prove the second part, we first observe using the first part that $\nsc(s,A) = \min\{1,s/g\} < 1/2$ if and only  if $s< g/2$. The result follows by combining this fact with Lem.~\ref{lem:GNUPreform}, and recalling that $\mathcal{T} = \mathcal{T}_s$ in the current setting (see~\eqref{def:Ts}).
\end{proof}

Thm.~\ref{thm:girthsparsity} reveals that when the dictionary is a graph incidence matrix, the maximum sparsity level that can always be recovered via $\ell_1$-minimization is determined by the girth of the graph. Therefore, for graphs with small girth, we expect recovery experiments to fail for certain signals of relatively small sparsity. On the other hand, although one can not recover all $s$-sparse signals for a given $s$, due to Cor.~\ref{cor:extremepts}, the successful recovery rate can still be relatively high. Because, even though the girth is small, if the graph consists mainly of large simple cycles, $\mathcal{T}_{\max}(A)$ can still contain a large number of index sets of cardinality $s$, for which $\ell_1$-recovery will always be successful. We will illustrate this intuition with experiments in \S~\ref{ssec:inc_matr_exp}.

\begin{rem}\label{rem:girth_poly_t}
Thm.~\ref{thm:girthsparsity} has an important consequence regarding the computational complexity of the calculation of nullspace constant for incidence matrices. We already mentioned in Rem.~\ref{rem:nscNP} that it is NP-hard to calculate the nullspace constant in general. However, there are algorithms that can calculate the girth of a graph exactly in $\mathcal{O}(mn)$. Therefore, Thm.~\ref{thm:girthsparsity} reveals that the nullspace constant can be calculated---hence the NUP can be verified---for graph incidence matrices in polynomial time.
\end{rem}

\subsection{Experiments on Graph Incidence Matrices}\label{ssec:inc_matr_exp}
All experiments are implemented in Python 3.6. In these experiments, we often run an $\ell_1$-recovery experiment that requires sampling a random vector $\bar{x}$ of a given sparsity level $s$. For this purpose we use the \texttt{random} function from Python's \texttt{scipy.sparse} class. The support of the sparse vector $\bar{x}$ is chosen uniformly at random, and its nonzero values are chosen from the standard normal distribution. The resulting vector is subsequently normalized to have unit $\ell_2$ norm, thus forming the final $\bar{x}$. For this $\bar{x}$, a vector $\hat{x}$ is then obtained as a solution to the optimization problem~\eqref{eq:l1relax}.
The attempted $\ell_1$-recovery is deemed successful if 
$\|\hat{x}-\bar{x}\|_2 \leq 10^{-6}$. 
It is well-known that  problem~\eqref{eq:l1relax} can
be cast as a linear program, which we solve using the \texttt{linprog} function in the \texttt{scipy.optimize} class.

In this section, we consider numerical experiments that are based on two different families of graph incidence matrices.

\subsubsection{Graphs from Fig.~\ref{fig:Exp_Inc_Matr_1_graphs}}\label{sssec:Exp_Inc_Matr_1}

Consider the family of graphs $\{\mathcal{G}_l\}$ depicted in Fig.~\ref{fig:Exp_Inc_Matr_1_graphs}. Graph $\mathcal{G}_l$ consists of a simple cycle of length three and a simple cycle of length $l+1$ for  $l\in \{3,5,7,9,11\}$. For each of these graphs, we tested the $\ell_1$-sparsity recovery capabilities by solving the problem
\begin{equation}\label{eq:Exp_Inc_Matr_1}
    \min_{x} \|x\|_1\text{ subject to }A(\mathcal{G}_l)\bar{x} = A(\mathcal{G}_l)x,
\end{equation}
where $x\in\mathbb{R}^{l+2}$ and $A(\mathcal{G}_l)$ denotes the incidence matrix of the graph $\mathcal{G}_l$.

For each pair of sparsity level $s$ and incidence matrix $A(\mathcal{G}_l)$, this experiment was repeated 
$2000$ times for different choices of $\bar{x}$,
and the ratio of the successful trials is reported in Fig.~\ref{fig:Exp_Inc_Matr_1_simple}.

Since the girth is 3 for each  $\mathcal{G}_l$ considered in this experiment, it follows from  Thm.~\ref{thm:girthsparsity} that only 1-sparse vectors can be recovered exactly, which is observed in Fig.~\ref{fig:Exp_Inc_Matr_1_simple}. Moreover, due to Cor.~\ref{cor:extremepts}(ii), we expect to have better exact recovery performance for graphs with larger simple cycles, which in our case means larger values of $l$. This increase in recovery probability can be observed in Fig.~\ref{fig:Exp_Inc_Matr_1_simple} for sparsity levels $s > 1$.

\begin{figure}
\centering
\subfloat[Graph $\mathcal{G}_l$ for $l \geq 3$.]{\includegraphics[width=0.2\textwidth]{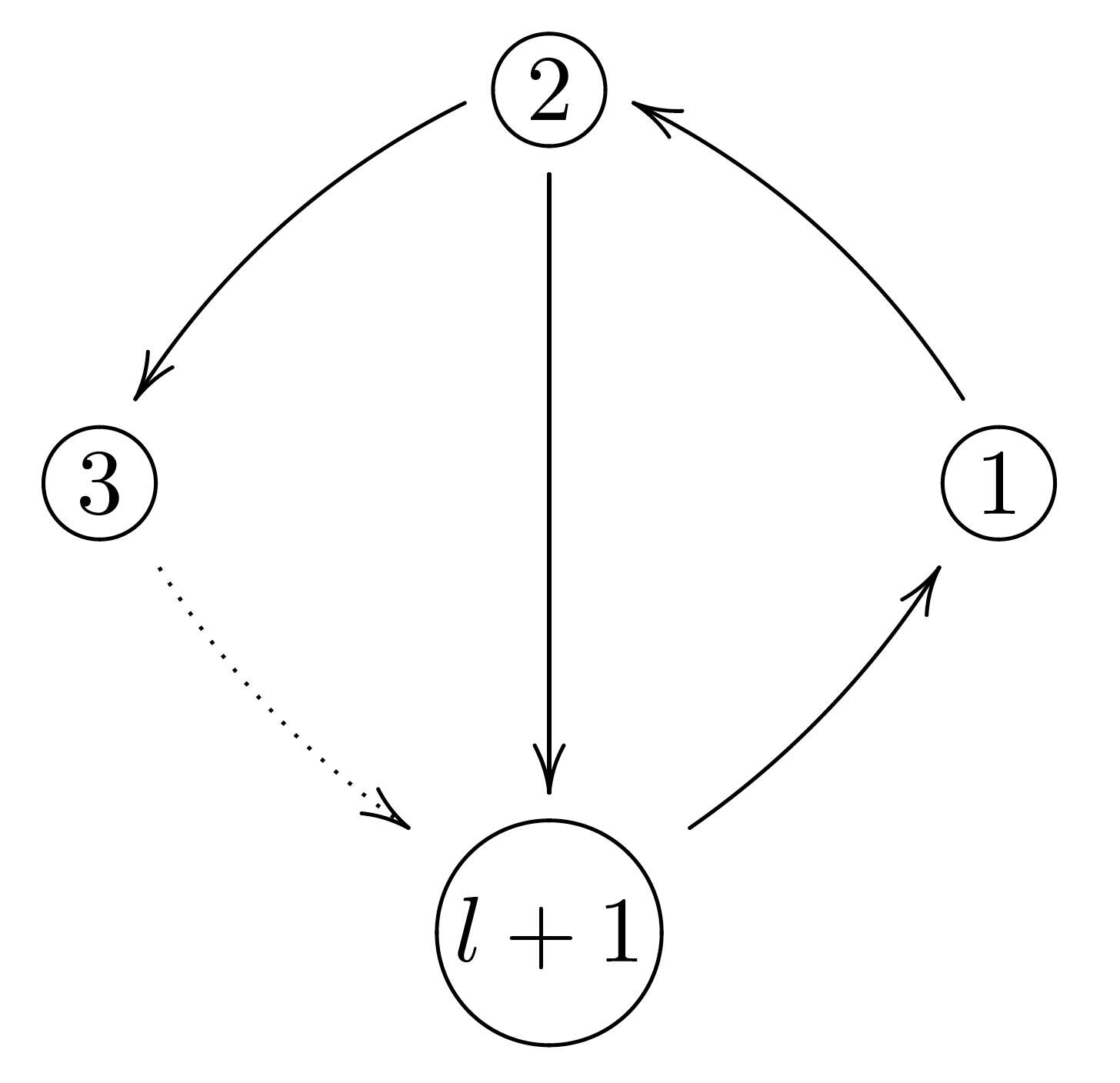}\label{fig:Exp_Inc_Matr_1_graphs}
}
~
\subfloat[Sparse recovery results for  $\mathcal{G}_l$.]{\includegraphics[width=0.27\textwidth]{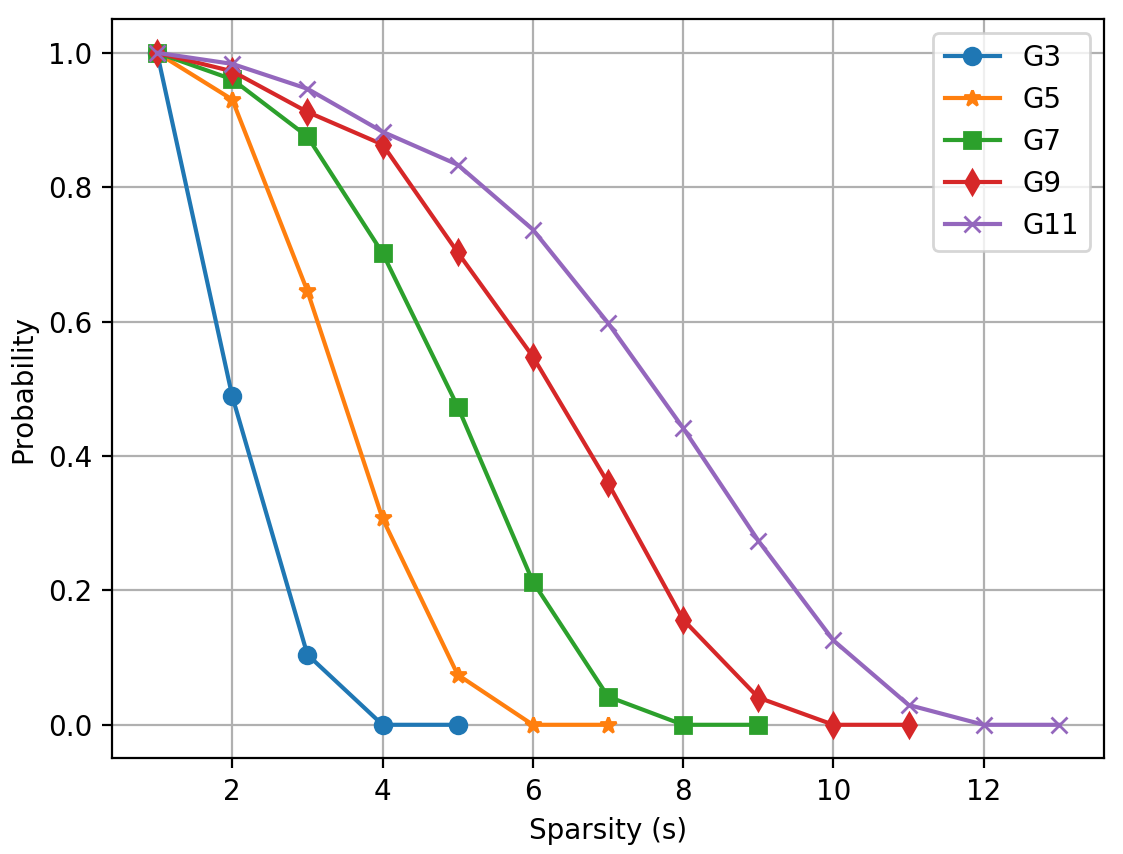}\label{fig:Exp_Inc_Matr_1_simple}}%
\caption{The graphs $\{\mathcal{G}_l\}$ used in \S\ref{sssec:Exp_Inc_Matr_1} and the sparse recovery results for $\mathcal{G}_l$ with $\ell\in\{3,5,7,9,11\}$ using \eqref{eq:Exp_Inc_Matr_1}.}
\end{figure}

\subsubsection{Erd\"{o}s-R\'{e}nyi Graphs}\label{sec:ER}
The previous experiment showed that the success of $\ell_1$-recovery for incidence matrices of deterministic graphs relies on the topology of the graph. The simple cycles of the graph completely determine the  recovery capability for signals with a given support. Now, 
we investigate the exact recovery probability for incidence matrices arising from randomly generated Erd\"{o}s-R\'{e}nyi graphs. 

Consider Erd\"{o}s-R\'{e}nyi graphs on 100 vertices.  
Any two vertices have an edge linking them with probability $\mathfrak{p}\in (0,1)$. For a fixed $\mathfrak{p}$, we sample 20 such graphs $\{\mathcal{G}_l\}_{0\leq l\leq 19}$, calculate their incidence matrices $\{A(\mathcal{G}_l)\}_{0\leq l\leq 19}$ 
and consider the $\ell_1$-recovery problem \eqref{eq:Exp_Inc_Matr_1}. 

For a given sparsity level $s$ and incidence matrix $A(\mathcal{G}_l)$, we sample 100 vectors $\bar{x}$ and determine whether $\ell_1$ sparse recovery using~\eqref{eq:Exp_Inc_Matr_1} is successful for each $\bar{x}$, as described in the first paragraph of this section. We record how many times  successful recovery occurs.

We repeat this experiment for different values of $\mathfrak{p}$. A critical value for $\mathfrak{p}$ is $\mathfrak{p}_{\text{crit}}:=\frac{\ln(100)}{100}\simeq 0.046$. When $\mathfrak{p} < \mathfrak{p}_{\text{crit}}$ the random graph is almost surely disconnected, and when  $\mathfrak{p}> \mathfrak{p}_{\text{crit}}$ the random graph is almost surely connected. To highlight the influence of $\mathfrak{p}_{\text{crit}}$, we select $\mathfrak{p} \in \{\mathfrak{p}_{\text{crit}}^{k/9}: 1\leq k\leq 10\}$.

The result of this experiment is shown in Fig.~\ref{fig:Exp_Inc_Matr_2}. For a given $\mathfrak{p}$, the expected number of edges in an Erd\"{o}s-R\'{e}nyi graph with 100 vertices is $\binom{100}{2}\mathfrak{p} = 4950\mathfrak{p}$, and the expected number of cycles of length three is $\binom{100}{3}\mathfrak{p}^3 = 161700\mathfrak{p}^3$. For instance, for $\mathfrak{p}=\mathfrak{p}_{\text{crit}}$, the expected number of edges is approximately $228$ and the expected number of triangles is approximately 15, which explains the high successful recovery rate observed in Fig.~\ref{fig:Exp_Inc_Matr_2} for $\mathfrak{p}=\mathfrak{p}_{\text{crit}}$. On the other hand, when $\mathfrak{p}=\mathfrak{p}_{\text{crit}}^{1/9}\simeq 0.71$, the expected number of triangles is approximately 57960, which is much larger than the expected number of edges (approximately~3516). This explains the extremely poor recovery rates observed in Fig.~\ref{fig:Exp_Inc_Matr_2}.

\begin{figure}
\centering
\includegraphics[width=0.45\textwidth]{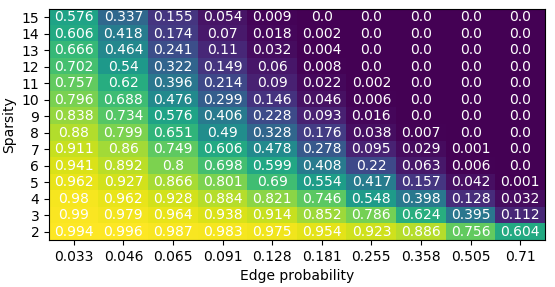}
\caption{Exact recovery probability of $\ell_1$-minimization via~\eqref{eq:Exp_Inc_Matr_1} when the dictionary is the incidence matrix of an Erd\"{o}s-R\'{e}nyi graph with $100$ vertices. The horizontal axis shows the edge probability $\mathfrak{p}$ and the vertical axis is the sparsity of the vector $\bar x$ that is attempted to be recovered. The critical edge probability value as described in Section~\ref{sec:ER} is  $\mathfrak{p}_{\text{crit}}\simeq 0.046$.}
\label{fig:Exp_Inc_Matr_2}
\end{figure}

\section{Application to Partial DFT Matrices}
The second application we present comes from the Discrete Fourier Transform (DFT) of dimension $n$, which we denote by $\mathcal{F}_n$. We consider the partial DFT matrices, each denoted by $\mathcal{F}_\Omega$ for some $\Omega\subseteq \mathcal{U}_n$, which are the submatrices of the DFT matrix consisting of the rows indexed by $\Omega$ with $m:=|\Omega|\geq 1$. 
Our goal is to understand the exact recovery capabilities of
\begin{equation}\label{eq:l1pDFT}
    \min \|x\|_1 \text{ subject to }\mathcal{F}_\Omega \bar{x} = \mathcal{F}_\Omega x \text{ and }x,\bar{x}\in\mathbb{R}^n.
\end{equation}

\subsection{Notation and Preliminaries for DFT matrices}
Before we present the second application in detail, we would like to make a few remarks about the notation. Up until this section we only dealt with real matrices and real vector spaces. Here we generalize some of the previous concepts to the complex domain. For a complex vector $c\in \mathbb{C}^n$, $\|c\|_0$ denotes the number of nonzero entries of the vector $c$. If $\Phi\in \mathbb{C}^{m\times n}$, then $\nullsp(\Phi)$ denotes the complex nullspace. That is, $\nullsp(\Phi) = \{z\in \mathbb{C}^n: \Phi z = 0\}$. In this subsection, when we talk about the dimension of a subspace $\mathcal{V}\subseteq\mathbb{C}^n$, we mean its dimension as a complex subspace. For any $\Omega, \Gamma \subseteq \mathcal{U}_n$ and any matrix $\Phi\in\mathbb{C}^{m\times n}$, we use $\Phi_{\Omega,\Gamma}$ to denote the submatrix of $\Phi$ consisting of the rows and columns from $\Phi$ whose indices are in $\Omega$ and $\Gamma$, respectively. For a complex number $z\in \mathbb{C}$, we let $|z|$ denote the norm of $z$.

For a positive integer $n$ and $\xi := e^{-\frac{2\pi}{n}i}$, the 
DFT matrix $\mathcal{F}_n\in \mathbb{C}^{n\times n}$ is defined entrywise as 
\begin{equation}
    [\mathcal{F}_n]_{k,l} := \tfrac{1}{\sqrt{n}}(\xi^{kl})
    \ \ \text{for} \ \ 0\leq k,l \leq n-1.
\end{equation}
The following remark gives useful properties of the DFT matrix in the case that $n$ is a prime number. 

\begin{rem}\label{rem:tao_k+1}
Tao~\cite{tao:03} shows that when $n$ is a prime number,  any $k\times k$ minor of  $\mathcal{F}_n$ is nonzero \cite[Lem.~1.3]{tao:03} for all $1 \leq k \leq n$. Also, as a consequence of \cite[Thm.~1.1]{tao:03}, any sparse polynomial
$\sum_{j=0}^{k}c_j z^{n_j}$
with $k + 1$ nonzero coefficients and $0 \leq n_0 < \cdots < n_k < n$, can have at most $k$ of the $n^{th}$ roots of unity as zeros, 
i.e., such a polynomial cannot vanish at more than $k$ of the $n^{th}$ roots of unity. {For an analogous result where $n$ need not be a prime number, we refer to \cite{meshulam:06}.}
\end{rem}

This remark has the following immediate consequences. 
\begin{lemma}\label{lem:tao-conseq}
Let $n$ be a prime number and $\Omega, \Gamma \subseteq \mathcal{U}_n$ be nonempty index sets satisfying $|\Gamma| = |\Omega|+1$. 
\begin{itemize}
    \item[(i)]
    All nonzero $c\in \nullsp(\mathcal{F}_{\Omega, \Gamma})$ satisfy $\|c\|_0 = |\Gamma|$.
    \item[(ii)]
    It holds that $\dim (\nullsp(\mathcal{F}_{\Omega,  \Gamma})) =1$.  Moreover, the nonzero vector $\nu\in \mathbb{C}^{|\Gamma|}$ defined entrywise as
\begin{equation*}
    \nu_t := (-1)^t\det(\mathcal{F}_{\Omega, \Gamma\setminus\{\Gamma_{t}\}}) 
    \ \ \text{for all} \  t\in \mathcal{U}_{|\Gamma|}
\end{equation*}
spans $\nullsp(\mathcal{F}_{\Omega, \Gamma})$.
\end{itemize}
\end{lemma}
\begin{proof}
To prove part (i) let $c\in \nullsp(\mathcal{F}_{\Omega, \Gamma})\setminus\{0\}$. If $\|c\|_0< |\Gamma|$, then $\sum_{k\in \Gamma} c_k z^k$ is a polynomial with $\|c\|_0\leq |\Omega|$ nonzero coefficients that vanishes at $|\Omega|$ roots of unity. This contradicts Tao's result in Rem.~\ref{rem:tao_k+1}, so that $\|c\|_0=|\Gamma|$ holds.

To prove the first part of (ii), we use an argument by contradiction.  Suppose that $\dim(\nullsp(\mathcal{F}_{\Omega, \Gamma})) > 1$, in which case there exists two linearly independent vectors $\{\bar{c}, \hat{c}\}\subset \nullsp(\mathcal{F}_{\Omega, \Gamma})$, which from part (i) must satisfy $\|\bar{c}\|_0 = \|\hat{c}\|_0 = |\Gamma|$. Next, define $\tilde{c} := (1/\bar{c}_1) \bar{c} - (1/\hat{c}_1)\hat{c}$ so that $\tilde{c} \in \nullsp(\mathcal{F}_{\Omega, \Gamma})$, $\tilde{c}_1 = 0$, and $\tilde{c} \neq 0$, which contradicts part (i).

To prove the second part of (ii), we note that for any $l\in \Omega$, the determinant of the $|\Gamma|\times |\Gamma|$ matrix 
$
\begin{bmatrix}
\mathcal{F}_{\{l\}, \Gamma}^T &
\mathcal{F}_{\Omega, \Gamma}^T
\end{bmatrix}^T
$
is zero because it has a repeated row. However, if we expand the determinant w.r.t. the first row we see that 
$$0=\sum_{t=0}^{|\Gamma|-1}(-1)^t \mathcal{F}_{\{l\}, \{\Gamma_t\}} \det(\mathcal{F}_{\Omega, \Gamma\setminus\{\Gamma_{t}\}}).$$
Since this holds for all $l\in \Omega$, we get $\nu \in \nullsp(\mathcal{F}_{\Omega, \Gamma})$. Moreover, it follows from the definition of $\nu$ that every entry in $\nu$ is nonzero because all square minors of $\mathcal{F}_n$ are nonzero (see Rem.~\ref{rem:tao_k+1}). This completes the proof.
\end{proof}

\subsection{Characterization of the MASC for Partial DFT Matrices}\label{ssec:pDFT}
We will characterize the exact recovery capabilities of \eqref{eq:l1pDFT} when $n$ is a prime number. This is achieved by characterizing
$\mathcal{T}_{\max}(\mathcal{F}_\Omega)$. 
By definition of the MASC, $\mathcal{T}_{\max}(\mathcal{F}_\Omega)$ is still a collection of supports of certain real signals, not complex ones. In fact, an index set $S\in \mathcal{T}_{\max}(\mathcal{F}_{\Omega})$ if and only if $\|z_S\|_1 < \frac{1}{2}$ for all $z\in \MinSupp(\nullsp(\mathcal{F}_\Omega)\cap\mathbb{R}^n)\cap\mathbb{S}^{n-1}_1$.
Therefore, our intermediate goal is to provide a description of the vectors of minimal support in $\nullsp(\mathcal{F}_\Omega)$, i.e. $\MinSupp(\nullsp(\mathcal{F}_\Omega)\cap\mathbb{R}^n)$. First note that if we define $\bar{\Omega} = \Omega \cup \{n-k: k\in \Omega\setminus\{0\}\}$, then $\nullsp(\mathcal{F}_\Omega)\cap \mathbb{B}_1^n = \nullsp(\mathcal{F}_{\bar{\Omega}})\cap \mathbb{B}_1^n$, which holds because if a polynomial with real coefficients vanishes at an $n^{th}$ root of unity $\xi^{k}$, then it also vanishes at its complex conjugate $\xi^{n-k}$.
Thus, {for the purposes of the $\ell_1$-minimization problem \eqref{eq:l1pDFT}}, without loss of generality, we can assume that
\begin{equation}\label{eq:omegaAssump}
    k\in \Omega \implies n-k\in \Omega\text{ for }k=1,\dots,n-1
\end{equation}
since the  unknown signal is assumed to be real. 

\begin{thm}\label{thm:minSuppFomega}
Let $n$ be a prime number and $\Omega\subseteq \mathcal{U}_n$ be nonempty index set satisfying \eqref{eq:omegaAssump}.  Then the following hold:
\begin{itemize}
    \item[(i)] If $\Gamma \subseteq\mathcal{U}_n$
satisfies $|\Gamma| = |\Omega|+1$, then $\nullsp(\mathcal{F}_{\Omega, \Gamma})$ is spanned by a real vector.  Moreover, it must be the case that  $\nullsp(\mathcal{F}_\Omega)
\cap \nullsp(\proj_{\Gamma^c})
\cap \mathbb{S}^{n-1}_1
\neq \emptyset$.
    \item[(ii)] The $\ell_1$ unit norm minimum support vectors satisfy \begin{align*}
    \MinSupp(&\nullsp(\mathcal{F}_\Omega)\cap\mathbb{R}^n) \cap \mathbb{S}_1^{n-1}  \\
    &= \!\!\!\!\!\! \bigcup_{\substack{\Gamma\subseteq \mathcal{U}_n\\ |\Gamma| = |\Omega|+1}} \!\!\!\!\!\!
    \big(\nullsp(\mathcal{F}_\Omega)
    \cap \nullsp(\proj_{\Gamma^c})
    \cap \mathbb{S}^{n-1}_1\big).
\end{align*}
\end{itemize}    
\end{thm}
\begin{proof}
To prove part (i), we start by letting $\Gamma\subseteq\mathcal{U}_n$ satisfy $|\Gamma| = |\Omega|+1$.
It follows from Lem.~\ref{lem:tao-conseq}(ii) that
$\dim (\nullsp(\mathcal{F}_{\Omega, \Gamma})) = 1$, so we can find a nonzero $\nu\in\mathbb{C}^{|\Gamma|}$ that spans $\nullsp(\mathcal{F}_{\Omega, \Gamma})$. Since $\Omega$ satisfies \eqref{eq:omegaAssump}, the complex conjugate  of $\nu$ is also in the nullspace, i.e., $\bar{\nu}\in \nullsp(\mathcal{F}_{\Omega, \Gamma})$. If all entries of $\nu$ are purely imaginary, then we can divide all entries of $\nu$ by  $i$ to obtain a nonzero real vector $\nu_{\text{real}}$ that spans $\nullsp(\mathcal{F}_{\Omega, \Gamma})$. Otherwise, the vector  $\nu_{\text{real}} := \nu + \bar{\nu}$ is nonzero, real, and spans $\nullsp(\mathcal{F}_{\Omega, \Gamma})$. Using $\nu_{\text{real}}\in\mathbb{R}^{|\Gamma|}$, we then define
$\bar{\nu}_{\text{real}}\in \mathbb{R}^{n}$ via $[\bar{\nu}_{\text{real}}]_{\Gamma} 
= \nu_{\text{real}}/\|\nu_{\text{real}}\|_1$
and 
$[\bar{\nu}_{\text{real}}]_{\Gamma^c} 
= 0$,
so that $\bar{\nu}_{\text{real}}\in \nullsp(\mathcal{F}_\Omega)
\cap \nullsp(\proj_{\Gamma^c})
\cap \mathbb{S}^{n-1}_1$, completing the proof.

To prove part (ii), we show both inclusions in turn.

[$\supseteq$]
Let $z\in \nullsp(\mathcal{F}_\Omega)
    \cap \nullsp(\proj_{\Gamma^c})
    \cap \mathbb{S}_1^{n-1}$
for some $\Gamma\subseteq\mathcal{U}_n$
satisfying $|\Gamma| = |\Omega|+1$. In particular, this means that $z \neq 0$,  $z\in\nullsp(\mathcal{F}_\Omega)\cap\mathbb{R}^n$, and $\supprt(z)\subseteq\Gamma$ .  We also claim that $z\in\MinSupp(\nullsp(\mathcal{F}_\Omega)\cap\mathbb{R}^n)$, which can be seen as follows.  Suppose, in order to reach a contradiction, that $z$ is not in $\MinSupp(\nullsp(\mathcal{F}_\Omega)\cap\mathbb{R}^n)$, i.e., that there exists a nonzero vector $\bar{z}$ satisfying $\supprt(\bar{z})\subsetneq \supprt(z)\subseteq\Gamma$ and $\bar{z}\in\nullsp(\mathcal{F}_\Omega)\cap\mathbb{R}^n$. It follows that $c:= [\bar{z}]_\Gamma$ satisfies $c\in\nullsp(\mathcal{F}_{\Omega,\Gamma})\setminus\{0\}$ and $\|c\|_0 < \Gamma$, which contradicts Lem.~\ref{lem:tao-conseq}(i). Thus, $z\in \MinSupp(\nullsp(\mathcal{F}_\Omega)\cap\mathbb{R}^n) \cap \mathbb{S}_1^{n-1}$, as claimed.    

[$\subseteq$] 
Let $z\in\MinSupp(\nullsp(\mathcal{F}_\Omega)\cap\mathbb{R}^n)\cap \mathbb{S}_1^{n-1}$ and define its support as $\Gamma = \supprt(z)$.
We claim that  $|\Gamma| \leq |\Omega|+1$. 
For a proof by contradiction, suppose that $|\Gamma| > |\Omega|+1$.  Then choose any $\bar\Gamma\subset \Gamma$ satisfying $|\bar{\Gamma}| = |\Omega| + 1$, and then note that $\nullsp(\mathcal{F}_{\Omega,\bar{\Gamma}})$ is nontrivial so that there exists a vector $\bar{\eta}\in\mathbb{R}^{|\bar{\Gamma}|}$ satisfying $\bar{\eta} \in\nullsp(\mathcal{F}_{\Omega,\bar{\Gamma}})\cap\mathbb{S}_1^{|\bar{\Gamma}|-1}$. Defining $\bar{z}\in\mathbb{R}^n$ as $[\bar{z}]_{\bar{\Gamma}} = \bar{\eta}$ and $[\bar{z}]_{\bar{\Gamma}^c} = 0$, it follows that $\bar{z}\in\mathbb{S}^{n-1} \cap \nullsp(\mathcal{F}_\Omega)$ and   $\supprt(\bar{z})\subseteq\bar{\Gamma} \subsetneq\Gamma = \supprt(z)$, which contradicts $z\in\MinSupp(\nullsp(\mathcal{F}_\Omega)\cap\mathbb{R}^n)$. Thus, we know that $|\Gamma| \leq |\Omega|+1$ holds. 
Finally, choose any $\hat{\Gamma}\subseteq\mathcal{U}_n$ satisfying $\Gamma \subseteq \hat{\Gamma}$ and $|\hat{\Gamma}| = |\Omega| + 1$. It follows that $z \in
\nullsp(\mathcal{F}_\Omega)
    \cap \nullsp(\proj_{\hat{\Gamma}^c})
    \cap \mathbb{S}^{n-1}_1$
with
$\hat{\Gamma}\subseteq\mathcal{U}_n$
and $|\hat{\Gamma}| = |\Omega| + 1$,
thus completing the proof.
\end{proof}

Thm.~\ref{thm:minSuppFomega} provides an interesting  correspondence between  $\MinSupp(\nullsp(\mathcal{F}_\Omega)\cap\mathbb{R}^n) \cap \mathbb{S}_1^{n-1}$ and the index sets $\Gamma\subseteq \mathcal{U}_n$ with $|\Gamma| = |\Omega|+1$. A consequence of Thm.~\ref{thm:minSuppFomega} and Prop.~\ref{prop:minSuppExt} is 
\begin{equation*}
    \Ext(\nullsp(\mathcal{F}_\Omega)\cap \mathbb{B}_1^{n}) = \!\!\!\!\!\! \bigcup_{\substack{\Gamma\subseteq \mathcal{U}_n\\ |\Gamma| = |\Omega|+1}} \!\!\!\!\!\!  \big( \nullsp(\mathcal{F}_\Omega)
    \cap \nullsp(\proj_{\Gamma^c})
    \cap \mathbb{S}^{n-1}_1\big)
\end{equation*}
anytime $\nullsp(\mathcal{F}_\Omega)$ is nontrivial. 
{We are now ready to state the main result of this section.

\begin{thm}\label{thm:pDFTasc_0}
Let $n$ be a prime number, $\Omega\subseteq \mathcal{U}_n$ satisfy \eqref{eq:omegaAssump}, and $S\subseteq \mathcal{U}_n$.  Then, the following statements are equivalent
\begin{enumerate}
    \item $S\in \mathcal{T}_{\max}(\mathcal{F}_\Omega)$.\label{item:pDFTasc_0_S}
    \item For all $\Gamma\subseteq \mathcal{U}_n$ with $|\Gamma| = |\Omega|+1$, we have
$$ \sum_{k\in S\cap \Gamma} |\det\left(\mathcal{F}_{\Omega, \Gamma\setminus\{k\}}\right)| <  \sum_{k\in S^{c}\cap \Gamma} |\det\left(\mathcal{F}_{\Omega, \Gamma\setminus\{k\}}\right)|.$$\label{item:pDFTasc_0_Gamma}
    \item Every $\bar{x}\in \mathbb{R}^n$ with $\supprt(\bar{x})\subseteq S$ is the unique solution to \eqref{eq:l1pDFT}.\label{item:pDFTasc_0_l1}
\end{enumerate}
\end{thm}
\begin{proof}
The equivalence of \eqref{item:pDFTasc_0_S} and \eqref{item:pDFTasc_0_Gamma} follows from  Prop.~\ref{prop:maxASCminSupp},
Thm.~\ref{thm:minSuppFomega}(ii), and
Lem.~\ref{lem:tao-conseq}(ii). The equivalence of \eqref{item:pDFTasc_0_S} and \eqref{item:pDFTasc_0_l1} follows from Thm.~\ref{thm:genthm2}.
\end{proof}}

{
We discuss the implications of Thm.~\ref{thm:pDFTasc_0} with an example.
\begin{exm}\label{exm:pDFT_exm}
Let $n=11$ and consider a $5\times 11$ partial DFT matrix, $\mathcal{F}_{\Omega}$, where $\Omega = \{0,2,4,7,9\}$. From Thm.~\ref{thm:minSuppFomega}, a vector $x\in \nullsp(\mathcal{F}_\Omega)\cap \mathbb{R}^n$ has minimal support if and only if $\|x\|_0=5+1=6$. Furthermore, by Thm.~\ref{thm:pDFTasc_0}, we have $S\in \mathcal{T}_{\max}(\mathcal{F}_\Omega)$ if and only if for all $\Gamma\subseteq \mathcal{U}_{11}$ with $|\Gamma|=6$, we have
$$\sum_{k\in S\cap \Gamma} |\det\left(\mathcal{F}_{\Omega, \Gamma\setminus\{k\}}\right)| <  \sum_{k\in S^{c}\cap \Gamma} |\det\left(\mathcal{F}_{\Omega, \Gamma\setminus\{k\}}\right)|.$$
Then, a simple computer program can be coded to check this condition, which reveals that MASC is given by
\begin{align*}
\mathcal{T}_{\max}(\mathcal{F}_\Omega) = &\left\{S\in 2^{\mathcal{U}_{11}}: |S|\leq 1\right\}\\
                                        &\cup \left\{\{\alpha,\beta\}\in 2^{\mathcal{U}_{11}}: \left||\alpha-\beta|-\frac{11}{2}\right|\geq \frac{3}{2}\right\},
\end{align*}
which has $56$ elements (including the empty set). The  recovery over $1$-sparse signals is uniform. However, $2$-sparse signals are not recovered uniformly. For instance, $\{0,5\}\notin \mathcal{T}_{\max}(\mathcal{F}_\Omega)$. Hence, there are signals $\bar{x}\in \mathbb{R}^{11}$ with $\supprt(\bar{x})=\{0,5\}$ which cannot be recovered via \eqref{eq:l1pDFT}. Furthermore, $\mathcal{T}_{\max}(\mathcal{F}_\Omega)$ does not contain any index set of cardinality $3$ or higher.
\end{exm}
}

The determinants appearing in
Thm.~\ref{thm:pDFTasc_0} are a consequence of the computation of a vector spanning the one dimensional subspace  $\nullsp(\mathcal{F}_{\Omega, \Gamma})$ when $|\Gamma| = |\Omega|+1$, as indicated by Lem.~\ref{lem:tao-conseq}(ii).
Unfortunately, this characterization of the  $\nullsp(\mathcal{F}_{\Omega, \Gamma})$ is not useful for computational purposes since it requires the computation of a possibly large number of determinants of large matrices. 
One approach for significantly reducing the computation is to make the following assumption.
\begin{ass}\label{ass:Omega}
Let $n$ be a prime number that defines the size of the DFT matrix $\mathcal{F}_n$. 
The index set $\Omega$ takes the form
\begin{equation}\label{eq:assump_first_mbar}
    \Omega = \{0, 1,\dots, \bar{m}, n-\bar{m}, \dots, n-1\}
\end{equation}
for some $\bar{m}\in\{1,2,\dots, (n-3)/2 \}$
so that $|\Omega| = 2\bar{m}+1 \leq n-2$ since $p$ is a prime number. 
This implies that any $\Gamma\subseteq\mathcal{U}_n$ satisfying $|\Gamma| = |\Omega|+1$ must also satisfy $|\Gamma| \leq n-1$. 
\end{ass}

Given $\Omega$ satisfying  Assumption~\ref{ass:Omega}, we define, for any $\Gamma\subseteq\mathcal{U}_n$ satisfying $|\Gamma| = |\Omega| + 1$, the polynomial 
\begin{equation}\label{def:fGamma}
    f_\Gamma(z) 
    := \prod_{k\in \Gamma} (z-\xi^k)
    \equiv \sum_{i=0}^{n-1} c_i(\Gamma) z^i
\end{equation}
for some vector $c(\Gamma)\in\mathbb{R}^n$ of the form
\begin{equation}\label{def:cGamma}
c(\Gamma) 
= [ \
c_0(\Gamma) \
c_1(\Gamma) \ 
\dots \
c_{|\Gamma|}(\Gamma) \
0 \
0 \ 
\dots \
0 \
]^T.
\end{equation}
The action of the DFT on $c(\Gamma)$ is given componentwise by
\begin{equation}\label{DFT-on-c}
[\mathcal{F}_n c(\Gamma)]_k
= \sum_{i=0}^{n-1} c_i(\Gamma) \xi^{ik}
= \sum_{i=0}^{n-1} c_i(\Gamma) (\xi^{k})^i
= f_\Gamma(\xi^k), 
\end{equation}
which is a useful property. It follows from \eqref{def:fGamma} and \eqref{DFT-on-c} that 
\begin{equation} \label{F-is-zero}
[\mathcal{F}_n c(\Gamma)]_\Gamma = 0.
\end{equation}
The definition in \eqref{def:fGamma} also allows us to rewrite the determinants needed in Thm.~\ref{thm:pDFTasc_0} in a convenient form, as we now show. 

\begin{prop}\label{prop:extr_pt_fprime_char}
Let $n$ be a prime number and $\Omega\subseteq \mathcal{U}_n$ satisfy \eqref{eq:assump_first_mbar}.  If $\Gamma\subseteq \mathcal{U}_n$ satisfies $|\Gamma| = |\Omega|+1$, then 
$$
\det\left(\mathcal{F}_{\Omega, \Gamma\setminus\{k\}}\right) = \pm\beta(\Gamma)\frac{1}{\xi^{k(n-\bar{m})} f^\prime_\Gamma(\xi^{k})} \ \ \text{for all $k\in\Gamma$},
$$
where $f^\prime_\Gamma(z)$ is the derivative of $f_\Gamma(z)$ and 
\begin{equation}\label{def:betaGamma}
\beta(\Gamma) :=
\xi^{(n-\bar{m})\sum_{l\in \Gamma} l}
\prod_{l\in\Gamma,r\in \Gamma, l<r}\!\!\!\!(\xi^l-\xi^r) 
\end{equation}
is a complex constant independent of $k$.
\end{prop}
\begin{proof}
Let $k\in\Gamma$ and define $\bar{\Gamma} := \Gamma\setminus \{k\}$. 
If, for each $l\in\bar{\Gamma}$, we divide the corresponding column of $\mathcal{F}_{\Omega, \bar{\Gamma}}$ 
by the constant $\xi^{l(n-\bar{m})}$, then we end up with a Vandermonde matrix, $V(\bar{\Gamma})\in \mathbb{C}^{(2\bar{m}+1)\times(2\bar{m}+1)}$, whose entries are of the form
$[V(\bar{\Gamma})]_{i,j} = (\xi^{\bar{\Gamma}_j})^i$.
Using properties of determinants and the formula for the determinant of the Vandermonde matrix, we get
\begin{align*}
    \det(\mathcal{F}_{\Omega, \bar{\Gamma}})
        &=  \xi^{(n-\bar{m})\sum_{l\in \bar{\Gamma}} l}
        \prod_{l\in\bar{\Gamma},r\in \bar{\Gamma}, l<r}\!\!\!\!(\xi^l-\xi^r).
\end{align*}
Note that, up to a sign change, we have
\begin{equation*}
    \prod_{l\in\Gamma,r\in \Gamma, l<r}\!\!\!\!\!\!(\xi^l-\xi^r) = \pm\!\!\left(\prod_{l\in\bar{\Gamma},r\in \bar{\Gamma}, l<r}\!\!\!\!\!\!(\xi^l-\xi^r)\right)\!\!\left(\prod_{l\in \Gamma, l\neq k}\!\!\!\!(\xi^k-\xi^l)\right)\!.
\end{equation*}
But since
\begin{equation*}
    \prod_{l\in \Gamma, l\neq k}(\xi^k-\xi^l) = f^\prime_\Gamma(\xi^{k}),
\end{equation*}
up to a sign change, we have 
\begin{equation*}
    \det(\mathcal{F}_{\Omega, \bar{\Gamma}})
        = \pm\frac{\xi^{(n-\bar{m})\sum_{l\in \Gamma} l} \prod_{l\in\Gamma,r\in \Gamma, l< r}(\xi^l-\xi^r)}{\xi^{k(n-\bar{m})} f^\prime_\Gamma(\xi^{k})}.
\end{equation*}
The result now follows from the definition of $\beta(\Gamma)$ in~\eqref{def:betaGamma}. 
\end{proof}

We can further develop Prop.~\ref{prop:extr_pt_fprime_char} by observing that
\begin{align*}
nz^{n-1}&= \frac{d}{dz}(z^n-1) = \frac{d}{dz}(f_\Gamma(z)f_{\Gamma^c}(z)) \\
        &= f^\prime_\Gamma(z)f_{\Gamma^c}(z) + f_\Gamma(z)f_{\Gamma^c}^\prime(z),
\end{align*}
which combined with $f_\Gamma(\xi^k)=0$ for all $k\in\Gamma$, gives
\begin{equation*}
f'_\Gamma(\xi^k)
= \frac{n \xi^{k(n-1)}}{f_{\Gamma^c}(\xi^k)}\text{ for all }k\in\Gamma.
\end{equation*}
Combining this identity with 
Prop.~\ref{prop:extr_pt_fprime_char} gives the following result.

\begin{prop}\label{prop:extr_pt_f_char}
Let $n$ be a prime number and $\Omega\subseteq \mathcal{U}_n$ satisfy \eqref{eq:assump_first_mbar}.  
If $\Gamma\subseteq \mathcal{U}_n$ satisfies  $|\Gamma| = |\Omega|+1$, then
$$
\det\left(\mathcal{F}_{\Omega, \Gamma\setminus\{k\}}\right) = \pm\beta(\Gamma)\frac{f_{\Gamma^c}(\xi^k)}{n\xi^{(n-\bar{m}-1)k}}
\ \ \text{for all $k\in \Gamma$},
$$
where $\beta(\Gamma)\in \mathbb{C}$ is the constant defined in~\eqref{def:betaGamma}.
\end{prop}

With the help of these propositions, we can provide another characterization of $\mathcal{T}_{\max}(\mathcal{F}_\Omega)$ when $\Omega$ satisfies \eqref{eq:assump_first_mbar} and the dimension of the unknown real signal is a prime number.

\begin{cor}\label{cor:pDFTasc}
Let $n$ be a prime number and the set $\Omega\subseteq \mathcal{U}_n$ satisfy \eqref{eq:assump_first_mbar}. It follows that $S\subseteq \mathcal{U}_n$ is in 
$\mathcal{T}_{\max}(\mathcal{F}_\Omega)$ if and only if for all $\Gamma\subseteq \mathcal{U}_n$ satisfying $|\Gamma| = |\Omega|+1$, one of the following equivalent conditions holds:
\begin{enumerate}[label=(C\arabic*), ref=C\arabic*, leftmargin=3\parindent]
    \item $\sum_{k\in S\cap \Gamma} |f^\prime_\Gamma(\xi^k)|^{-1} < \sum_{k\in S^c\cap \Gamma} |f^\prime_\Gamma(\xi^k)|^{-1}$. \label{eq:fPrimeGamma}
    \item $\sum_{k\in S\cap\Gamma}|f_{\Gamma^c}(\xi^k)| < \sum_{k\in S^c\cap\Gamma}|f_{\Gamma^c}(\xi^k)|$. \label{eq:fGammaC}
\end{enumerate}
\end{cor}
\begin{proof}
These results follow by combining  Thm.~\ref{thm:pDFTasc_0} with
Prop.~\ref{prop:extr_pt_fprime_char} (for (C1)) and 
Prop.~\ref{prop:extr_pt_f_char} (for (C2)).
\end{proof}
The identity in \eqref{DFT-on-c} shows that \eqref{eq:fGammaC} can be interpreted as a condition on the DFT acting on the vectors $c(\Gamma^c)$ defined in~\eqref{def:cGamma}. Hence, existing results can be used to bound the level of sparsity for which exact recovery is possible. To take advantage of this, we require the following definition.

\begin{defn}[coherence]\label{def:coherence}
Let $\Phi\in \mathbb{C}^{p\times q}$ have $\ell_2$-normalized columns. The coherence of $\Phi$ is defined as
$$
\mu(\Phi) := \max_{k\neq l} |\Phi_k^* \Phi_l|,
$$
where $\Phi_k$ denotes the $k$th column of $\Phi$.
\end{defn}

The following known result uses the coherence to show that a nonzero vector cannot be arbitrarily well concentrated with respect to two different orthonormal bases.

\begin{lemma}[\cite{RiBo:18}]\label{lem:RiBoLemma}
Let $A,B \in \mathbb{C}^{n\times n}$ be unitary and $S,T \subseteq \mathcal{U}_n$ be nonempty. If there exists a nonzero  $\bar{z} \in \mathbb{C}^n$ satisfying $\|\bar{z} - \proj_S(\bar{z})\|_1\leq \varepsilon \|\bar{z}\|_1$ for some $\varepsilon \in [0,1]$, and a nonzero $\bar{y} \in \mathbb{C}^n$ with $\supprt(\bar{y}) = T$ such that $A\bar{z} = B\bar{y}$, then
\begin{equation}\label{ST-bound}
    |S||T|
    \geq \frac{1 - \varepsilon}{\mu^2(\begin{bmatrix} A & B\end{bmatrix})},
\end{equation}
where the \emph{coherence} $\mu(\cdot)$ is defined in Defn.~\ref{def:coherence}.
\end{lemma}
Using this lemma we can get a bound on the sparsity of signals that can be recovered exactly via \eqref{eq:l1pDFT}.
\begin{cor}\label{cor:pDFT_rec_s_sparse}
Let $n$ be a prime number and $\Omega\subseteq \mathcal{U}_n$ satisfy \eqref{eq:assump_first_mbar}. Every $s$-sparse signal $\bar{x}\in \mathbb{R}^n$ is the unique solution to the optimization problem in~\eqref{eq:l1pDFT} 
if 
\begin{equation} \label{s-bound}
s < \tfrac{n}{2(n-|\Omega|)}.
\end{equation}
\end{cor}
\begin{proof}
Define $A := I_{n}\in\mathbb{C}^{n\times n}$ (the $n\times n$ identity matrix) and $B:=\mathcal{F}_n\in\mathbb{C}^{n\times n}$, and then 
observe using Defn.~\ref{def:coherence} that 
\begin{equation}\label{coherence-equals}
\mu(\begin{bmatrix} I_n & \mathcal{F}_n\end{bmatrix}) = 1/ \sqrt{n}.
\end{equation}
Also note that since $\Omega$ is assumed to satisfy~\eqref{eq:assump_first_mbar}, we know that $|\Omega| = 2\bar{m} + 1$ for some $\bar{m}\in\{1,2,\dots,(n-3)/2\}$.

Next, let $T:= \mathcal{U}_{n-2\bar{m}-1}$ and let $S\subseteq\mathcal{U}_n$ be any set satisfying $|S| \leq s$. Combining this with~\eqref{s-bound} we find that
$$
|S|
\leq s
< \tfrac{n}{2(n-|\Omega|)}
= \tfrac{n}{2(n-2\bar{m} - 1)}
= \tfrac{n}{2|T|},
$$
which combined with~\eqref{coherence-equals} and choosing $\varepsilon = 1/2$ yields
$$
|S||T|
< \tfrac{n}{2}
= \tfrac{1}{2\mu^2([I_n \, \mathcal{F}_n])}
= \tfrac{1-\varepsilon}{\mu^2([I_n \, \mathcal{F}_n])}.
$$
Since this shows that~\eqref{ST-bound} does not hold, we must conclude from Lem.~\ref{lem:RiBoLemma} that
there does not exist a pair of nonzero vectors $\bar{y},\bar{z} \in \mathbb{C}^n$ satisfying $\|\bar{z} - \proj_{S}(\bar{z})\|_1\leq \tfrac{1}{2}\|\bar{z}\|_1$,
$\supprt(\bar{y}) = T$,
and
$\bar{z} =\mathcal{F}_n\bar{y}$; this is equivalent to saying that all vectors $\bar{y}\in\mathbb{C}^n$ with $\supprt(\bar{y}) = T$ must satisfy $
\|\mathcal{F}_n\bar{y} - \proj_{S}(\mathcal{F}_n\bar{y})\|_1
> \tfrac{1}{2}\|\mathcal{F}_n\bar{y}\|_1,
$
which itself may equivalently be stated as
\begin{equation} \label{key:T-bound}
\|(\mathcal{F}_n\bar{y})_{S^c}\|_1
> \|(\mathcal{F}_n\bar{y})_S\|_1
\ \ \text{$\forall \bar{y}$ with $\supprt(\bar{y}) = T$.} 
\end{equation}

Next, consider any $\Gamma\subseteq \mathcal{U}_n$ satisfying $|\Gamma| = |\Omega| + 1 =2(\bar{m} + 1)$, which implies that $|\Gamma^c| = n - 2(\bar{m}+1)$. It follows from this equality, \eqref{def:cGamma}, and the definition of $T$ that 
$$
\supprt(c(\Gamma^c))
= \mathcal{U}_{|\Gamma^c|+1}
= \mathcal{U}_{n-2\bar{m}-1}
= T.
$$
Thus, using~\eqref{DFT-on-c}, \eqref{F-is-zero}, and 
$\bar{y} = c(\Gamma^c)$  in~\eqref{key:T-bound}, we have
\begin{align*}
\sum_{k\in S^c\cap\Gamma}|f_{\Gamma^c}(\xi^k)|
&= \sum_{k\in S^c\cap\Gamma}|[\mathcal{F}_n(c(\Gamma^c))]_k| \\ &= \sum_{k\in S^c}|[\mathcal{F}_n(c(\Gamma^c))]_k|  \\
&= \|(\mathcal{F}_n c(\Gamma^c))_{S^c}\|_1 
 > \|(\mathcal{F}_n c(\Gamma^c))_S\|_1 \\
&= \sum_{k\in S}|[\mathcal{F}_n(c(\Gamma^c))]_k| \\
&= \sum_{k\in S\cap\Gamma}|[\mathcal{F}_n(c(\Gamma^c))]_k|
= \sum_{k\in S\cap\Gamma}|f_{\Gamma^c}(\xi^k)|,
\end{align*}
so that \eqref{eq:fGammaC} holds.
It now follows from Cor.~\ref{cor:pDFTasc} that $S\in \mathcal{T}_{\max}(\mathcal{F}_\Omega)$, i.e., $\{S\} \subset \mathcal{T}_{\max}(\mathcal{F}_\Omega)$, which in turn implies via Thm.\ref{thm:genthm2} that any $\bar{x}$ satisfying $\supprt(\bar{x}) = S$ is the unique solution to the optimization problem~\eqref{eq:l1pDFT}. Since $S$  satisfied $|S| \leq s$ but was otherwise arbitrary, the proof is complete.
\end{proof}

From Cor.~\ref{cor:pDFT_rec_s_sparse}, the partial DFT matrices we consider in Assumption~\ref{ass:Omega}, coherence based bounds on maximum sparsity that can be recovered are extremely simple to calculate. We note this as the biggest merit of this corollary. More generally, the results we provide in this subsection show how the machinery we built in the previous section can be used to describe the collection of all support sets for which $\ell_1$-recovery is always successful, when the dictionary is the partial Fourier transform - probably the most commonly studied dictionary in the compressed sensing/sparse recovery literature.

\subsection{Experiments on Partial DFT Matrices}\label{ssec:pDFT_exp}

In this section we perform two numerical experiments aimed at better understanding the theoretical results obtained for sparse recovery when the dictionary is a partial DFT matrix. 
We refer to \S~\ref{ssec:inc_matr_exp} for implementation details, in particular how the sampling of each random vector $\bar{x}$ is performed.

\subsubsection{Tightness of the Bound in Cor.~\ref{cor:pDFT_rec_s_sparse}}

The bound in Cor.~\ref{cor:pDFT_rec_s_sparse} on the sparsity for which recovery is guaranteed is not expected to be tight. We illustrate this with a simple experiment using a partial DFT matrix $\mathcal{F}_\Omega$ with $n=19$ and $\Omega$ chosen to satisfy \eqref{eq:assump_first_mbar} with $\bar{m}=7$ so that $|\Omega|=15$. By Cor.~\ref{cor:pDFT_rec_s_sparse}, we know that if $s < 19/8 = 2.375$, then every $s$-sparse signal $\bar{x}\in \mathbb{R}^{19}$ can be recovered as a solution of the $\ell_1$-minimization problem 
\begin{equation}\label{eq:pDFT19}
        \min_{x}
        \|x\|_1 \text{ subject to }\mathcal{F}_\Omega \bar{x}= \mathcal{F}_\Omega x. 
\end{equation}
Thus, the corollary guarantees recovery only for $s\in \{1,2\}$. 

To test this conclusion, for each sparsity level $s$, we sampled 1000 signals $\bar{x}\in \mathbb{R}^{19}$ and tested whether $\ell_1$-recovery via solving~\eqref{eq:pDFT19} was successful.  
The probability of recovering the true signal for each sparsity level is illustrated in  Fig.~\ref{fig:Exp_pDFT_1a}. In that same figure, we also plot the probability that a support set of cardinality $s$ is in $\mathcal{T}_{\max}(\mathcal{F}_\Omega)$, which is calculated using Cor.~\ref{cor:pDFTasc}. (We terminate the plot once the probabilities hit zero.) From this second plot, we see that the highest sparsity level that can be recovered is three, thus  showing that the bound provided by Cor.~\ref{cor:pDFT_rec_s_sparse} is not tight. Also, although the probability of an index set being in $\mathcal{T}_{\max}(\mathcal{F}_\Omega)$ drops rapidly, the drop rate for the exact recovery probability using  $\ell_1$-minimization is much slower. This indicates that for a given support set $S\notin\mathcal{T}_{\max}(\mathcal{F}_\Omega)$, the signals with support $S$ that are recoverable via $\ell_1$-minimization can still be abundant. {Moreover, this plot gives a summary of the structure of $\mathcal{T}_{\max}(\mathcal{F}_\Omega)$ in terms of support sets of cardinality $s$.}

\begin{figure}[t]
\centering
\includegraphics[width=0.4\textwidth]{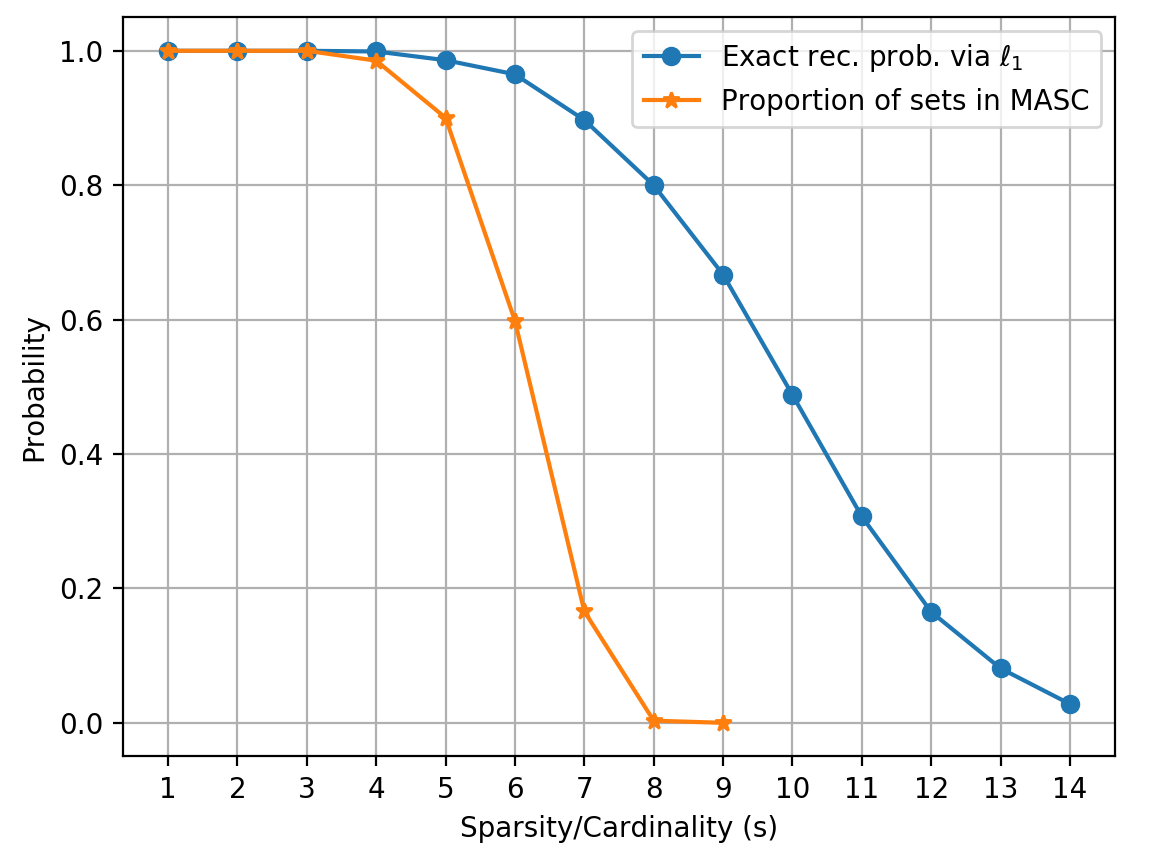}
\caption{The exact recovery probability (blue line) via $\ell_1$-minimization using~\eqref{eq:pDFT19} 
when the signal $\bar{x}$ is real-valued of dimension $n=19$ and the measurement set $\Omega$ for the partial DFT matrix $\mathcal{F}_\Omega$ satisfies  \eqref{eq:assump_first_mbar} with $\bar{m} = 7$. We also plot (orange line) the  probability that a randomly selected support set of cardinality $s$ from $\mathcal{U}_n$ is in $\mathcal{T}_{\max}(\mathcal{F}_\Omega)$.}
\label{fig:Exp_pDFT_1a}
\end{figure}

\subsubsection{The Maximal Recoverable Sparsity Level (MRSL)}\label{sec:dft-2}

Let $n$ be a prime number.  For any index set $\Omega\subseteq \mathcal{U}_n$ that satisfies \eqref{eq:assump_first_mbar},
what is the maximal sparsity level for which every vector of that sparsity level can be recovered via $\ell_1$-minimization? We call this number the MRSL and denote it by  $\smax(\Omega)$.
Cor.~\ref{cor:pDFT_rec_s_sparse} gives  a lower bound for $\smax(\Omega)$. In principle, $\smax(\Omega)$ can be obtained  by computing every extreme point of $\nullsp(\mathcal{F}_\Omega)\cap \mathbb{B}_1^n$ and checking whether  \eqref{eq:fGammaC} in Cor.~\ref{cor:pDFTasc} is satisfied. More precisely, let us first define
\begin{equation*}
    \mathfrak{E}(\Omega) := \{\Gamma\subseteq \mathcal{U}_n: |\Gamma|=|\Omega|+1\}.
\end{equation*}
For each $\Gamma\in \mathfrak{E}(\Omega)$, we sort  $\{|f_{\Gamma^c}(\xi^k)|\}_{k\in\mathcal{U}_n}$ into decreasing order and define $\smax(\Omega,\Gamma)$ as the maximum integer 
so that the sum of the first $\smax(\Omega,\Gamma)$ largest elements of $\{|f_{\Gamma^c}(\xi^k)|\}_{k\in\mathcal{U}_n}$ does not exceed half of the total sum $\sum_{k=0}^{n-1} |f_{\Gamma^c}(\xi^k)|$, so that  \eqref{eq:fGammaC} in Cor.~\ref{cor:pDFTasc} holds. It follows that
\begin{equation}\label{eq:smax}
    \smax(\Omega) \equiv \min_{\Gamma\in \mathfrak{E}(\Omega)} \smax(\Omega,\Gamma).
\end{equation}
For large $n$, the cardinality of $\mathfrak{E}(\Omega)$ (i.e., the number of extreme points)  is  
huge, making the calculation in \eqref{eq:smax} impractical. This motivates the need for an upper bound on $\smax(\Omega)$ that is tractable. We describe such a strategy next.

Let $\widehat{\mathfrak{E}}(\Omega)$ denote any randomly selected subset of  $\mathfrak{E}(\Omega)$.  For such a choice, it follows from~\eqref{eq:smax} that
\begin{equation}\label{eq:smax_estimate}
\widehat{\smax}(\Omega):=\min_{\Gamma\in \widehat{\mathfrak{E}}(\Omega)} \smax(\Omega,\Gamma)
    \geq \smax(\Omega),
\end{equation}
which means that  $\widehat{\smax}(\Omega)$ is an upper bound for $\smax(\Omega)$.  Moreover, the computation of $\widehat{\smax}(\Omega)$ will be efficient provided we select $\widehat{\mathfrak{E}}(\Omega)$ to be a relatively small subset of $\mathfrak{E}(\Omega)$.
To evaluate this strategy, we use the following set of matrices.

\begin{test}\label{test1}
Let $n=1009$ (a prime number) and choose the collection of sets $\{\Omega_j\}_{j=0}^{38}\subset \mathcal{U}_{1009}$ to satisfy \eqref{eq:assump_first_mbar} with sizes  $|\Omega_j| = 247+20j$ for each $0 \leq j \leq 38$.\footnote{To illustrate how large the number of extreme points can be, we remark that when $|\Omega| = 507$, the number of extreme points is $\genfrac(){0pt}{}{1009}{508} \simeq 1.34\times 10^{302}$.}
\end{test}

\begin{figure}[t]
\centering
\includegraphics[width=0.4\textwidth]{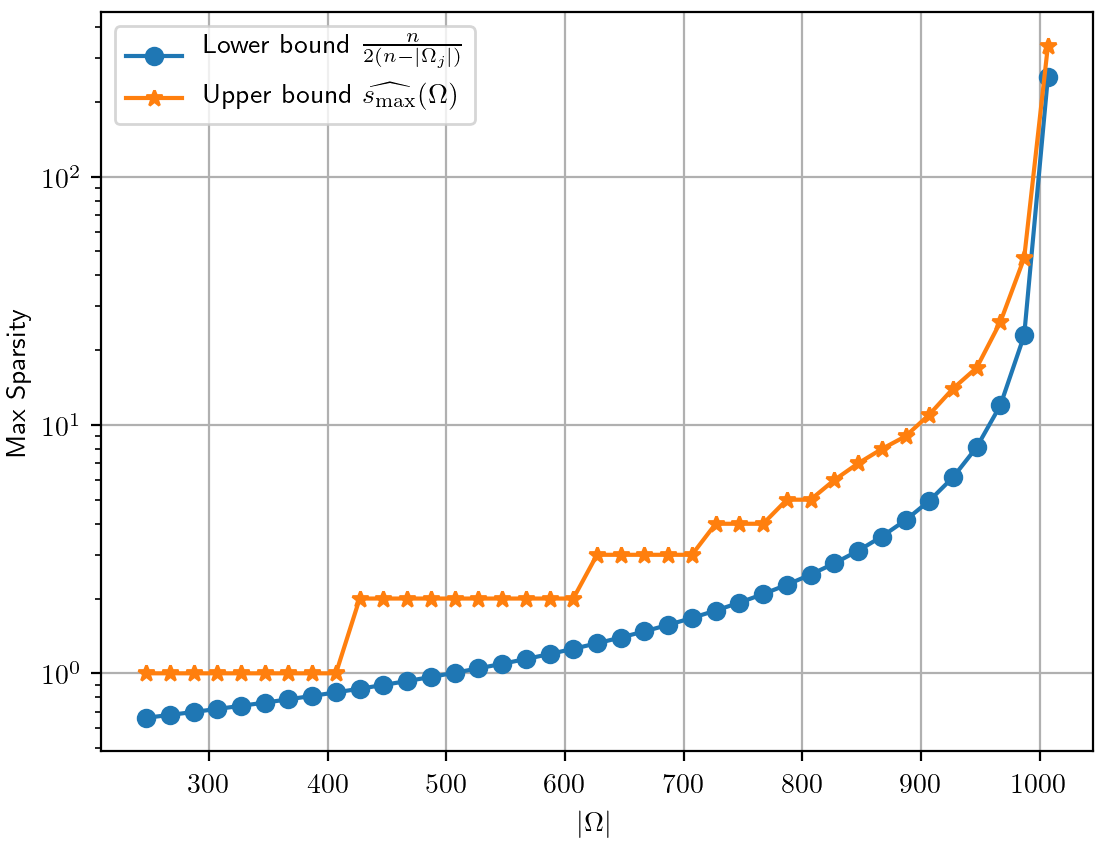}
\caption{Bounds for the MRSL  via \eqref{eq:pDFT19} for the subsets $\{\Omega_j\}_{j=0}^{38}$ in Test~\ref{test1}. 
The lower bound $\frac{n}{2(n-|\Omega_j|)}$ (blue line) is given by 
Cor.~\ref{cor:pDFT_rec_s_sparse} and the upper bound  $\widehat{\smax}(\Omega)$ by~\eqref{eq:smax_estimate} (orange line).}
\label{fig:Exp_pDFT_2a}
\end{figure}

For each $\Omega_j$ described in Test~\ref{test1}, we randomly select a subset $\widehat{\mathfrak{E}}(\Omega_j)\subset {e}(\Omega_j)$ of size $1000$ and proceed to compute $\widehat{\smax}(\Omega_j)$ using \eqref{eq:smax_estimate} as an upper bound on 
the MRSL when using the $\ell_1$-recovery problem   
\eqref{eq:pDFT19} with $\Omega = \Omega_j$.
The results are presented in Fig.~\ref{fig:Exp_pDFT_2a}. For each $\Omega_j$, we plot the sampling-based upper bound $\widehat{\smax}(\Omega_j)$ and the coherence-based lower bound  $\frac{n}{2(n-|\Omega_j|)}$ derived in Cor.~\ref{cor:pDFT_rec_s_sparse}. Note that for {$|\Omega_j|\leq 600$} it holds that  $\smax(\Omega_j) \leq 2$ so that the maximal  sparsity level that can be recovered is small.

To illustrate that our strategy is better than a na\"{i}ve approach that repetedly solves \eqref{eq:pDFT19} to estimate $s_{\max}(\Omega)$ (see Algorithm~\ref{alg:simple} below), consider the following setup.

\begin{test}\label{test2}
Let $n=61$ (a prime number) and define $\{\Omega_j\}_{j=0}^{22} \subset \mathcal{U}_n$ by choosing $\bar m\in\{7,8,\dots,29\}$ in~\eqref{eq:assump_first_mbar}, which means that 
$|\Omega_j| = 15+ 2j$ for $0 \leq j \leq 22$.
\end{test}
For each $\Omega_j$ described in Test~\ref{test2}, we use Algorithm~\ref{alg:simple} with inputs $\Omega = \Omega_j$ and $K = 1000$ to obtain $\widetilde\smax(\Omega_j)$ as an upper bound to $\smax(\Omega_j)$.
The results are shown in Fig.~\ref{fig:Exp_pDFT_2b}, where we also include the upper bound $\widehat{\smax}(\Omega)$ in \eqref{eq:smax_estimate} obtained via sampling of the extreme points. For large values of $|\Omega_j|$, the upper bound $\widehat{\smax}(\Omega_j)$ is significantly better. Moreover, in terms of the computation time the estimation obtained via sampling of the extreme points is more than an order of magnitude faster in comparison to the na\"{i}ve approach we described in Algorithm~\ref{alg:simple}.

\begin{figure}[t]
\centering
\includegraphics[width=0.4\textwidth]{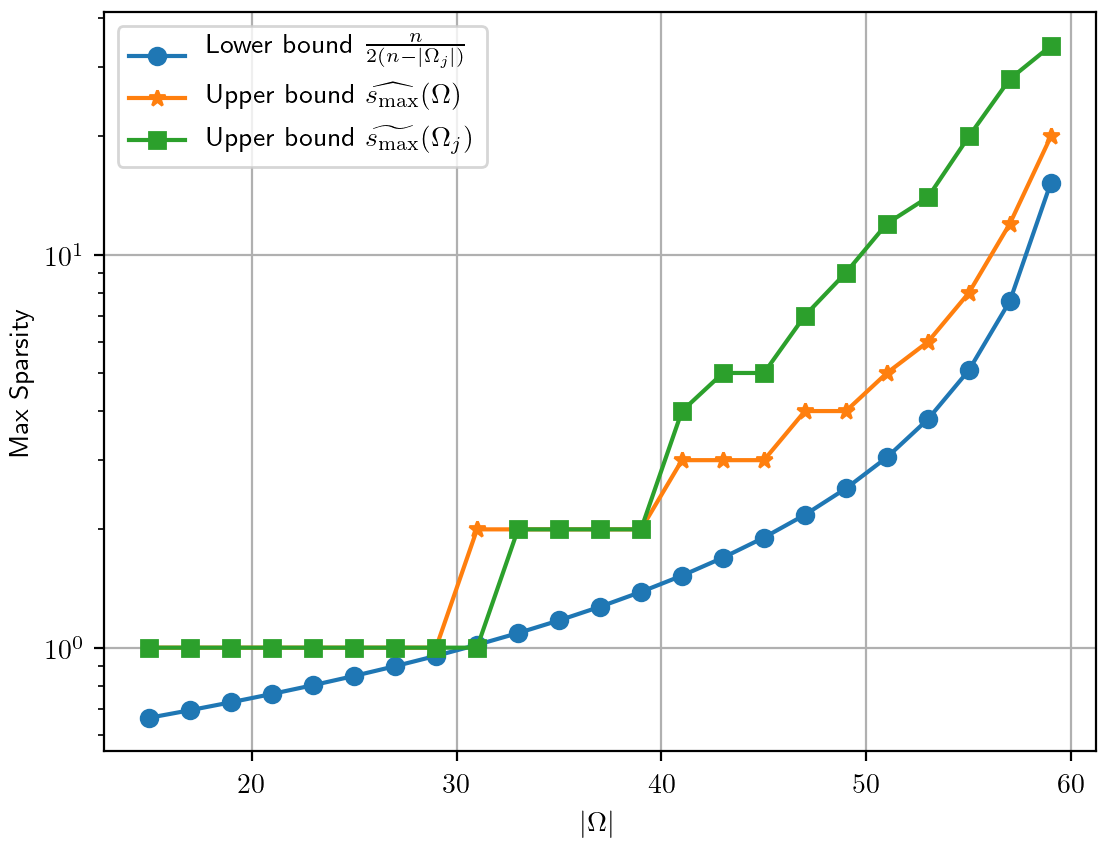}
\caption{Bounds for the MRSL via \eqref{eq:pDFT19} for the matrices $\{\Omega_j\}_{j=0}^{22}$ in Test~\ref{test2}. 
The lower bound $\frac{n}{2(n-|\Omega_j|)}$ (blue line) is given by Cor.~\ref{cor:pDFT_rec_s_sparse}. The upper bound $\widehat{\smax}(\Omega)$ (orange line) is  given by~\eqref{eq:smax_estimate} and the upper bound $\widetilde\smax(\Omega_j)$ (green line) is returned by Algorithm~\ref{alg:simple} with inputs $\Omega = \Omega_j$ and $K = 1000$.} 
\label{fig:Exp_pDFT_2b}
\end{figure}

\begin{algorithm}
\caption{A simple upper bound for the MRSL.}
\label{alg:simple}
\begin{algorithmic}[1]
  \State \textbf{inputs:} $\Omega \subset \mathcal{U}_n$ and sampling size $K \geq 1$.
  \State set $\tilde{s} \gets n$
  \Loop
  \State Randomly select $K$ vectors in $\mathbb{R}^n$ that are $\tilde{s}$-sparse.
  \If{all $K$ vectors are correctly recovered via~\eqref{eq:pDFT19}}
    \State \textbf{return} $\widetilde{\smax}(\Omega) \gets \tilde{s}$
  \EndIf
  \State set $\tilde{s} \gets \tilde{s} - 1$
  \EndLoop
\end{algorithmic}
\end{algorithm}

\section{Conclusion}
In this paper we introduced a new framework for studying the recovery of arbitrary sparsity patterns via $\ell_1$-minimization. We showed that there is a maximal recoverable sparsity pattern for each  dictionary $\Phi\in \mathbb{R}^{m\times n}$, which we called the \emph{maximum abstract simplicial complex (MASC) associated with $\Phi$}. We provided a characterization of the MASC using the extreme points of $\nullsp(\Phi)\cap \mathbb{B}_1^n$, and a second characterization using the vectors of minimal support of $\nullsp(\Phi)$. Furthermore, we showed how our approach benefits the analysis of sparse recovery when the dictionary is an incidence matrix associated with a simple graph or a $m\times p$ partial DFT matrix, where $p$ is a prime number {and the unknown signal is real}. In particular, we gave a complete characterization of the MASC associated with these matrix classes, which allowed us to characterize the collection of all support sets for which exact recovery via $\ell_1$-minimization is always possible. Interestingly, we showed that when the dictionary is an incidence matrix, the \emph{Nullspace Property} can be verified in polynomial time, although this condition is NP-hard to check for matrices in general. {We also showed that a computationally more advantageous characterization can be achieved if stronger assumptions are imposed on the measurement indices of the partial DFT matrix.}

Our framework opens the door for new directions of research. In particular, the connection between extreme points, vectors of minimal support, and the GNUP can be exploited to study sparse recovery problems for new classes of dictionaries.

\ifCLASSOPTIONcaptionsoff
  \newpage
\fi



\bibliographystyle{IEEEtran}
\bibliography{biblio/biblio,biblio/additional_biblio}
\end{document}

%% file: edits/edits.tex
\usepackage{ifthen} 
\newboolean{showcomments}
\setboolean{showcomments}{true}   
\usepackage{todonotes}

\definecolor{bleudefrance}{rgb}{0.19, 0.55, 0.91}
\definecolor{ao(english)}{rgb}{0.0, 0.5, 0.0}

\newcommand{\enrique}[1]{  \ifthenelse{\boolean{showcomments}}
{\todo[inline,color=bleudefrance]{Enrique says: #1}}{}}

\newcommand{\rene}[1]{\ifthenelse{\boolean{showcomments}}
{\todo[inline,color=yellow]{Rene says: #1}}{}}

\newcommand{\addcite}[0]{\ifthenelse{\boolean{showcomments}}
{\textcolor{purple}{(add cite(s)) }}{}}%

\newcommand{\emmargin}[1]{\ifthenelse{\boolean{showcomments}}
{\todo{Enrique: #1)}}{}}

\newboolean{showedits}
\setboolean{showedits}{true}
\usepackage[xcolor=bleudefrance]{changes}
\definechangesauthor[color=bleudefrance]{EM}
\newcommand{\aem}[1]{
\ifthenelse{\boolean{showedits}}
{\added[id=EM]{#1}}
{#1}
}
\newcommand{\chem}[2]{
\ifthenelse{\boolean{showedits}}
{\replaced[id=EM]{#1}{#2}}
{#1}
}
\newcommand{\dem}[1]{
\ifthenelse{\boolean{showedits}}
{\deleted[id=EM]{#1}}
{}
}

\usepackage{amsmath}
\usepackage{amsthm}
\usepackage{amssymb}

\usepackage{graphicx}
\usepackage[all]{xy}

\usepackage{mathrsfs}

\newtheorem{defn}{Definition}
\newtheorem{thm}{Theorem}
\newtheorem{prop}{Proposition}
\newtheorem{cor}{Corollary}
\newtheorem{lemma}{Lemma}
\newtheorem{exm}{Example}
\newtheorem{rem}{Remark}

\DeclareMathOperator{\nullsp}{Null}
\DeclareMathOperator{\supprt}{Supp}

\DeclareMathOperator{\Ext}{Ext}

\DeclareMathOperator{\proj}{Pr}
\DeclareMathOperator{\nsc}{nsc}

\DeclareMathOperator{\rinte}{rinte}
\DeclareMathOperator{\MinSupp}{MinSupp}

\DeclareMathOperator{\clo}{clo}
\DeclareMathOperator{\aff}{aff}